\documentclass[10pt,letterpaper,twoside,reqno]{myIEEEtran}

\newtheorem{thm}{Theorem}
\newtheorem{cor}[thm]{Corollary}
\newtheorem{lemma}[thm]{Lemma}
\newtheorem{prop}[thm]{Proposition}

\DeclareMathAlphabet{\mathsfsl}{OT1}{cmss}{m}{sl}

\newcommand{\lang}{\textit}

\newcommand{\term}{\emph}

\newcommand{\cnst}[1]{\mathrm{#1}}

	{\makebox{\phantom{}} \\ %
		\textsc{Input:}\begin{itemize}}%
	{\end{itemize}}

	{\textsc{Output:}\begin{itemize}}%
	{\end{itemize}}

	{\textsc{Procedure:}\begin{enumerate}}%
	{\end{enumerate}}

\renewcommand{\phi}{\varphi}

\newcommand{\eps}{\varepsilon}

\newcommand{\defby}{\overset{\mathrm{\scriptscriptstyle{def}}}{=}}
\newcommand{\half}{\tfrac{1}{2}}

\newcommand{\econst}{\mathrm{e}}
\newcommand{\iunit}{\mathrm{i}}

\newcommand{\onevct}{\mathbf{e}}

\newcommand{\Id}{\mathbf{I}}

\newcommand{\Rspace}[1]{\mathbb{R}^{#1}}
\newcommand{\Cspace}[1]{\mathbb{C}^{#1}}

\newcommand{\oneton}[1]{\left\llbracket {#1} \right\rrbracket}

\newcommand{\abs}[1]{\left\vert {#1} \right\vert}
\newcommand{\abssq}[1]{{\abs{#1}}^2}

\newcommand{\sgn}[1]{\operatorname{sgn}{#1}}
\newcommand{\real}{\operatorname{Re}}
\newcommand{\imag}{\operatorname{Im}}

\newcommand{\diff}[1]{\mathrm{d}{#1}}
\newcommand{\idiff}[1]{\, \diff{#1}}

\newcommand{\argmin}{\operatorname*{arg\; min}}

\newcommand{\Prob}[1]{\operatorname{\mathbb{P}}\left\{ {#1} \right\}}
\newcommand{\Expect}{\operatorname{\mathbb{E}}}

\newcommand{\vct}[1]{\bm{#1}}
\newcommand{\mtx}[1]{\bm{#1}}

\newcommand{\adj}{*}

\newcommand{\diag}{\operatorname{diag}}

\newcommand{\supp}[1]{\operatorname{supp}(#1)}

\newcommand{\restrict}[1]{\big\vert_{#1}}

\newcommand{\ip}[2]{\left\langle {#1},\ {#2} \right\rangle}
\newcommand{\absip}[2]{\abs{\ip{#1}{#2}}}
\newcommand{\abssqip}[2]{\abssq{\ip{#1}{#2}}}

\newcommand{\norm}[1]{\left\Vert {#1} \right\Vert}
\newcommand{\normsq}[1]{\norm{#1}^2}

\newcommand{\enorm}[1]{\norm{#1}_2}
\newcommand{\enormsq}[1]{\enorm{#1}^2}

\newcommand{\fnorm}[1]{\norm{#1}_{\mathrm{F}}}

\newcommand{\pnorm}[2]{\norm{#2}_{#1}}
\newcommand{\infnorm}[1]{\norm{#1}_{\infty}}

\newcommand{\triplenorm}[1]{\left\vert\!\left\vert\!\left\vert {#1} \right\vert\!\right\vert\!\right\vert}

\newcommand{\subjto}{\quad\text{subject to}\quad}

\newcommand{\atom}{\vct{\phi}}
\newcommand{\Fee}{\mtx{\Phi}}

\newcommand{\bigO}{{\rm O}}

\newcommand{\sinc}{\operatorname{sinc}}

\usepackage{amsmath}
\usepackage{amssymb}
\usepackage{bm}
\usepackage{mathrsfs}
\usepackage{stmaryrd}

\usepackage[pdftex]{graphicx}

\usepackage{color}

\usepackage{supertabular}

\usepackage{epstopdf}
\usepackage{url}

\begin{document}

\title{Beyond Nyquist: \\
Efficient Sampling of Sparse Bandlimited Signals}

\author{Joel A.~Tropp, \IEEEmembership{Member, IEEE},
Jason N.~Laska, \IEEEmembership{Student Member, IEEE},
Marco F.~Duarte, \IEEEmembership{Member, IEEE},
\\Justin K.~Romberg, \IEEEmembership{Member, IEEE},  and
Richard G.~Baraniuk, \IEEEmembership{Fellow, IEEE}
\thanks{Submitted: 30 January 2009. Revised: 12 September 2009.  A preliminary report on this work was presented by the first author at SampTA 2007 in Thessaloniki.}
\thanks{JAT was supported by ONR N00014-08-1-0883, DARPA/ONR N66001-06-1-2011
and N66001-08-1-2065, and NSF DMS-0503299.
JNL, MFD, and RGB were supported by DARPA/ONR N66001-06-1-2011 and
N66001-08-1-2065, ONR N00014-07-1-0936, AFOSR FA9550-04-1-0148, NSF
CCF-0431150, and the Texas Instruments Leadership University
Program.  JR was supported by NSF CCF-515632.
}}

\maketitle

\begin{abstract}
Wideband analog signals push contemporary analog-to-digital
conversion systems to their performance limits. In many
applications, however, sampling at the Nyquist rate is inefficient
because the signals of interest contain only a small number of
significant frequencies relative to the bandlimit, although the
locations of the frequencies may not be known a priori. For this
type of sparse signal, other sampling strategies are possible. This
paper describes a new type of data acquisition system, called a
\term{random demodulator}, that is constructed from robust, readily
available components.  Let $K$ denote the total number of
frequencies in the signal, and let $W$ denote its bandlimit in Hz.
Simulations suggest that the random demodulator requires just $\bigO( K
\log(W / K) )$ samples per second to stably reconstruct the signal.
This sampling rate is \emph{exponentially lower} than the Nyquist
rate of $W$ Hz.  In contrast with Nyquist sampling, one must use
nonlinear methods, such as convex programming, to recover the signal
from the samples taken by the random demodulator.  This paper provides
a detailed theoretical analysis of the system's performance that
supports the empirical observations.
\end{abstract}


\begin{center}
Dedicated to the memory of Dennis M.~Healy.
\end{center}

\section{Introduction}

\PARstart{T}{he} Shannon sampling theorem is one of the foundations
of modern signal processing.  For a continuous-time signal $f$ whose
highest frequency is less than $W/2$ Hz, the theorem suggests that
we sample the signal uniformly at a rate of $W$ Hz.  The values of
the signal at intermediate points in time are determined completely
by the \term{cardinal series}
\begin{equation} \label{eqn:cardinal}
f(t) = \sum\nolimits_{n \in \mathbb{Z}} f\left(\frac{n}{W} \right)
   \sinc\left( Wt - n \right).
\end{equation}
In practice, one typically samples the signal at a somewhat higher
rate and reconstructs with a kernel that decays faster than the
$\sinc$ function \cite[Ch.~4]{OSB99:Discrete-Time-Signal}.

This well-known approach becomes impractical when the bandlimit $W$
is large because it is challenging to build sampling hardware that
operates at a sufficient rate.  The demands of many modern
applications already exceed the capabilities of current technology.
Even though recent developments in analog-to-digital converter (ADC)
technologies have increased sampling speeds, state-of-the-art
architectures are not yet adequate for emerging applications, such
as ultrawideband and radar systems because of the additional
requirements on power consumption~\cite{LeSPMag05}.  The time has
come to explore alternative techniques
\cite{Hea05:Analog-to-Information}.

\subsection{The Random Demodulator}

In the absence of extra information, Nyquist-rate sampling is
essentially optimal for bandlimited signals
\cite{Lan67:Sampling-Data}.  Therefore, we must identify other
properties that can provide additional leverage.  Fortunately, in
many applications, signals are also {\em sparse}.  That is, the
number of significant frequency components is often much smaller
than the bandlimit allows.  We can exploit this fact to design new
kinds of sampling hardware.

This paper studies the performance of a new type of sampling
system---called a \term{random demodulator}---that can be
used to acquire sparse, bandlimited signals.
Figure~\ref{fig:block} displays a block
diagram for the system, and Figure~\ref{fig:demod} describes the
intuition behind the design.  In summary, we demodulate the signal
by multiplying it with a high-rate pseudonoise sequence, which
smears the tones across the entire spectrum.  Then we apply a
lowpass anti-aliasing filter, and we capture the signal by sampling
it at a relatively low rate.  As illustrated in
Figure~\ref{fig:signatures}, the demodulation process ensures that
each tone has a distinct signature within the passband of the
filter.  Since there are few tones present, it is possible to
identify the tones and their amplitudes from the low-rate samples.

\begin{figure}
\begin{center}
\includegraphics[width=.9\columnwidth]{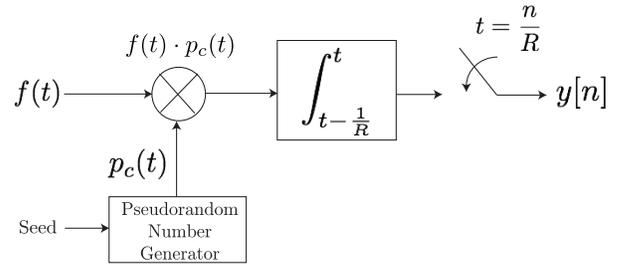}
\end{center}
\caption{{\sl Block diagram for the random demodulator.}  The components include a random number generator, a mixer, an accumulator, and a sampler.}
\label{fig:block}
\end{figure}

\begin{figure}
\begin{center}
\includegraphics[width=\columnwidth]{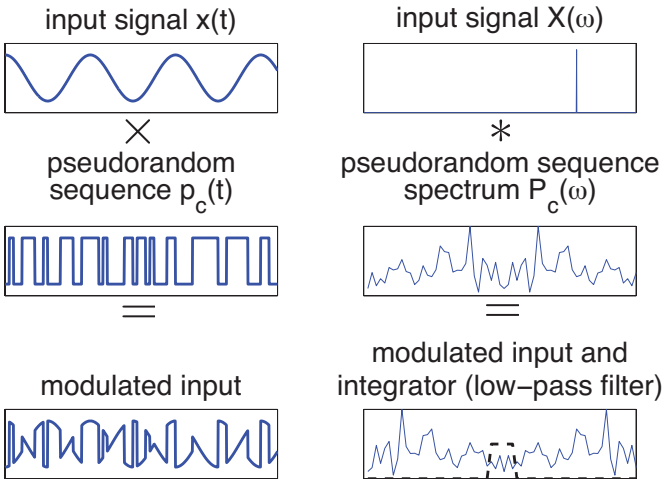}
\end{center}
\caption{{\sl Action of the demodulator on a pure tone.} The demodulation process multiplies the continuous-time input signal by a random square wave.  The action of the system on a single tone is illustrated in the time domain (left) and the frequency domain (right).  The dashed line indicates the frequency response of the lowpass filter.  See Figure~\ref{fig:signatures} for an enlargement of the filter's passband.} \label{fig:demod}
\end{figure}

\begin{figure}
\begin{center}
\includegraphics[width=\columnwidth]{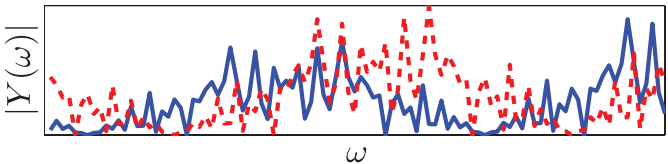}
\end{center}
\vspace*{-3mm} \caption{{\sl Signatures of two different tones.}
The random demodulator furnishes each
frequency with a unique signature that can be discerned by examining
the passband of the antialiasing filter.  This image enlarges the
pass region of the demodulator's output for two input tones (solid
and dashed). The two signatures are nearly orthogonal when their
phases are taken into account.} \label{fig:signatures}
\end{figure}

The major advantage of the random demodulator is that it bypasses
the need for a high-rate ADC. Demodulation is typically much easier
to implement than sampling, yet it allows us to use a low-rate ADC.
As a result, the system can be constructed from robust, low-power,
readily available components even while it can acquire
higher-bandlimit signals than traditional sampling hardware.

We do pay a price for the slower sampling rate: It is no longer
possible to express the original signal $f$ as a linear function of the
samples, \lang{{\`a} la} the cardinal series \eqref{eqn:cardinal}.
Rather, $f$ is encoded into the measurements in a more subtle
manner.  The reconstruction process is highly nonlinear, and must
carefully take advantage of the fact that the signal is sparse. As a
result, signal recovery becomes more computationally intensive. In
short, the random demodulator uses additional digital processing to
reduce the burden on the analog hardware. This tradeoff seems
acceptable, as advances in digital computing have outpaced those in
analog-to-digital conversion.

\subsection{Results}

Our simulations provide striking evidence that the random demodulator performs.
Consider a periodic signal with a bandlimit of $W/2$ Hz, and suppose
that it contains $K$ tones with random frequencies and phases.
Our experiments below show that, with high probability, the system
acquires enough information to reconstruct the signal after sampling
at just $\bigO( K \log(W/K) )$ Hz.  In words, the sampling rate is
proportional to the number $K$ of tones and the logarithm of the
bandwidth $W$.  In contrast, the usual approach requires sampling at
$W$ Hz, regardless of $K$. In other words, the random demodulator
operates at an \emph{exponentially} slower sampling rate!  We also
demonstrate that the system is effective for reconstructing simple
communication signals.

Our theoretical work supports these empirical conclusions, but it
results in slightly weaker bounds on the sampling rate.  We have
been able to prove that a sampling rate of $\bigO( K \log W  +
\log^3 W )$ suffices for high-probability recovery of the random
signals we studied experimentally.  This analysis also suggests that
there is a small startup cost when the number of tones is small, but
we did not observe this phenomenon in our experiments.  It remains
an open problem to explain the computational results in complete
detail.

The random signal model arises naturally in numerical experiments,
but it does not provide an adequate description of real signals,
whose frequencies and phases are typically far from random.  To
address this concern, we have established that the random demodulator can
acquire {\em all} $K$-tone signals---regardless of the frequencies,
amplitudes, and phases---when the sampling rate is $\bigO(K \log^6
W)$.  In fact, the system does not even require the spectrum of the
input signal to be sparse; the system can successfully recover any
signal whose spectrum is well-approximated by $K$ tones.  Moreover,
our analysis shows that the random demodulator is robust against
noise and quantization errors.

This work focuses on input signals drawn from a specific
mathematical model, framed in Section~\ref{sec:sigmod}.  Many real
signals have sparse spectral occupancy, even though they do not meet
all of our formal assumptions.  We propose a device, based on the
classical idea of \term{windowing}, that allows us to approximate
general signals by signals drawn from our model.  Therefore, our
recovery results for the idealized signal class extend to signals
that we are likely to encounter in practice.

In summary, we believe that these empirical and theoretical results,
taken together, provide compelling evidence that the demodulator
system is a powerful alternative to Nyquist-rate sampling for sparse
signals.

\subsection{Outline}

In Section~\ref{sec:sigmod}, we present a mathematical model for the
class of sparse, bandlimited signals.  Section~\ref{sec:rdm-demod}
describes the intuition and architecture of the random demodulator,
and it addresses the nonidealities that may affect its performance.
In Section~\ref{sec:demod-mtx}, we model the action of the random
demodulator as a matrix.  Section~\ref{sec:algs} describes
computational algorithms for reconstructing frequency-sparse signals
from the coded samples provided by the demodulator.  We continue
with an empirical study of the system in Section~\ref{sec:exp}, and
we offer some theoretical results in Section~\ref{sec:theory} that
partially explain the system's performance.
Section~\ref{sec:prewindowing} discusses a windowing technique that
allows the demodulator to capture nonperiodic signals.  We conclude
with a discussion of potential technological impact and related work
in Sections~\ref{sec:discussion} and Section~\ref{sec:related}.
Appendices~\ref{app:rand-demod-mtx},~\ref{app:random-phase}
and~\ref{app:rip} contain proofs of our signal reconstruction
theorems.

\section{The Signal Model} \label{sec:sigmod}

Our analysis focuses on a class of discrete, multitone signals that
have three distinguished properties:
\begin{itemize}
\item   {\bf Bandlimited.}  The maximum frequency is bounded.

\item   {\bf Periodic.}  Each tone has an integral frequency in Hz.

\item   {\bf Sparse.}   The number of active tones is small in comparison with the bandlimit.
\end{itemize}
Our work shows that the random demodulator can recover these signals very efficiently.  Indeed, the number of samples per unit time scales directly with the sparsity, but it increases only logarithmically in the bandlimit.

At first, these discrete multitone signals may appear simpler than the signals that arise in most applications. For example, we often encounter signals that contain nonharmonic tones or signals that contain continuous bands of active frequencies rather than discrete tones.  Nevertheless, these broader signal classes can be approximated within our model by means of \term{windowing} techniques.  We address this point in Section~\ref{sec:prewindowing}.

\subsection{Mathematical Model}

Consider the following mathematical model for a class of discrete
multitone signals.  Let $W/2$ be a positive integer that exceeds the
highest frequency present in the continuous-time signal $f$.
Fix a number $K$ that represents the number of active tones.
The model contains each signal of the form
\begin{equation} \label{eqn:cont-time}
f(t) = \sum\nolimits_{\omega \in \Omega}
   a_\omega \, \econst^{-2\pi \iunit \omega t}
\qquad\text{for $t \in [0, 1)$.}
\end{equation}
Here, $\Omega$ is a set of $K$ integer-valued frequencies that satisfies
$$
\Omega \subset \{ 0, \pm 1, \pm 2, \dots, \pm (W/2 - 1), W/2 \},
$$
and
$$
\{ a_\omega : \omega \in \Omega \}
$$
is a set of complex-valued amplitudes.  We focus on the case where the
number $K$ of active tones is much smaller than the bandwidth $W$.

To summarize, the signals of interest are {\em bandlimited} because they contain no frequencies above $W/2$ cycles per second; {\em periodic} because the frequencies are integral; and {\em sparse} because the number of tones $K \ll W$.

Let us emphasize several conceptual points about this model:
\begin{itemize}
\item   We have normalized the time interval to one second for simplicity.  Of course, it is possible to consider signals at another time resolution.

\item   We have also normalized frequencies.  To consider signals whose frequencies are drawn from a set equally spaced by $\Delta$, we would change the effective bandlimit to $W/\Delta$.

\item   The model also applies to signals that are sparse and bandlimited in a single time interval.  It can be extended to signals where the model \eqref{eqn:cont-time} holds with a different set of frequencies and amplitudes in each time interval $[0, 1), [1, 2), [2, 3), \dots$.  These signals are sparse not in the Fourier domain but rather in the short-time Fourier domain \cite[Ch.~IV]{Mal99:Wavelet-Tour}.
\end{itemize}

\subsection{Information Content of Signals}
\label{subsec:infocontent}

According to the sampling theorem, we can identify signals from the model \eqref{eqn:cont-time} by sampling for one second at $W$ Hz.  Yet these signals contain only $R = \bigO(K \log(W/K) )$ bits of information.  In consequence, it is reasonable to expect that we can acquire these signals using only $R$ digital samples.

Here is one way to establish the information bound.  Stirling's approximation shows that there are about $\exp\{ K \log(W/K) + \bigO(K) \}$ ways to select $K$ distinct integers in the range $\{1, 2, \dots, W\}$.  Therefore, it takes $\bigO( K \log (W/K) )$ bits to encode the frequencies present in the signal.  Each of the $K$ amplitudes can be approximated with a fixed number of bits, so the cost of storing the frequencies dominates.

\subsection{Examples}

There are many situations in signal processing where we encounter
signals that are sparse or locally sparse in frequency.
Here are some basic examples:
\begin{itemize}

\item{\bf Communications signals}, such as transmissions with a frequency hopping modulation scheme that switches a sinusoidal carrier among many frequency
channels according to a predefined (often pseudorandom) sequence.
Other examples include transmissions with narrowband modulation where the carrier frequency is unknown but could lie anywhere in a wide bandwidth.

\item{\bf Acoustic signals}, such as musical signals where each note
consists of a dominant sinusoid with a progression of several harmonic overtones.

\item{\bf Slowly varying chirps}, as used in radar and geophysics,
that slowly increase or decrease the frequency of a sinusoid over
time.

\item {\bf Smooth signals} that require only a few Fourier
coefficients to represent.

\item {\bf Piecewise smooth signals} that are differentiable except for
a small number of step discontinuities.
\end{itemize}

We also note several concrete applications where sparse
wideband signals are manifest.
Surveillance systems may acquire a broad swath of
Fourier bandwidth that contains only a few communications signals.
Similarly, cognitive radio applications rely on the fact that parts of
the spectrum are not occupied~\cite{CabMisBro::2004::Implementation-Issues},
so the random demodulator could be used to perform \emph{spectrum sensing}
in certain settings.
Additional potential applications include geophysical imaging, where
two- and three-dimensional seismic data can be modeled as piecewise
smooth (hence, sparse in a local Fourier representation)~
\cite{HerHen::2007::Non-parametric-seismic}, as well as radar and sonar
imaging~\cite{BarSte::2007::Compressive-radar,HerStr::2009::High-resolution-radar}.

\section{The Random Demodulator} \label{sec:rdm-demod}

This section describes the random demodulator system that we propose for signal acquisition.  We first discuss the intuition behind the system and its design.
Then we address some implementation issues and nonidealities that impact its performance.

\subsection{Intuition}

The random demodulator performs three basic actions: demodulation, lowpass filtering, and low-rate sampling. Refer back to Figure~\ref{fig:block} for the block diagram.
In this section, we offer a short explanation of why this approach allows us to acquire sparse signals.

Consider the problem of acquiring a single high-frequency tone that lies within a wide spectral band.  Evidently, a low-rate sampler with an antialiasing filter is oblivious to any tone whose frequency exceeds the passband of the filter.  The random demodulator deals with the problem by smearing the tone across the entire spectrum so that it leaves a signature that can be detected by a low-rate sampler.

More precisely, the random demodulator forms a (periodic) square wave that randomly alternates at or above the Nyquist rate.  This random signal is a sort of periodic approximation to white noise.
When we multiply a pure tone by this random square wave, we simply translate the spectrum of the noise, as documented in Figure~\ref{fig:demod}.  The key point is that translates of the noise spectrum look completely different from each other, even when restricted to a narrow frequency band, which Figure~\ref{fig:signatures} illustrates.

Now consider what happens when we multiply a frequency-sparse signal by the random square wave.  In the frequency domain, we obtain a superposition of translates of the noise spectrum, one translate for each tone. Since the translates are so distinct, each tone has its own signature.  The original signal contains few tones, so we can disentangle them by examining a small slice of the spectrum of the demodulated signal.

To that end, we perform lowpass filtering to prevent aliasing, and we sample with a low-rate ADC.  This process results in coded samples that contain a complete representation of the original sparse signal.  We discuss methods for decoding the samples in Section~\ref{sec:algs}.

\subsection{System Design}

Let us present a more formal description of the random demodulator shown in Figure~\ref{fig:block}.  The first two components implement the demodulation process.  The first piece is a random number generator, which produces a discrete-time sequence $\eps_0, \eps_1, \eps_2 \dots$ of numbers that take values $\pm 1$ with equal probability.  We refer to this as the \term{chipping sequence}.  The chipping sequence is used to create a continuous-time demodulation signal $p_c(t)$ via the formula
$$
p_c(t) = \eps_n, \quad\text{$t \in \biggl[\frac{n}{W},\frac{n+1}{W} \biggr)$ and $n = 0, 1, \dots, W - 1$.}
$$
In words, the demodulation signal switches between the levels $\pm 1$ randomly at the Nyquist rate of $W$ Hz. Next, the mixer multiplies the continuous-time input $f(t)$ by the demodulation signal $p_c(t)$ to obtain a continuous-time demodulated signal
$$
y(t) = f(t) \cdot p_c(t),
\quad\text{$t \in [0, 1)$.}
$$
Together these two steps smear the frequency spectrum of the original signal via the convolution
$$
Y(\omega) = (F * P_c)(\omega).
$$
See Figure~\ref{fig:demod} for a visual.

The next two components behave the same way as a standard ADC, which performs lowpass filtering to prevent aliasing and then samples the signal.  Here, the lowpass filter is simply an accumulator that sums the demodulated signal $y(t)$ for $1/R$ seconds.  The filtered signal is sampled instantaneously every $1/R$ seconds to obtain a sequence $\{y_m\}$ of measurements.  After each sample is taken, the accumulator is reset. In summary,
$$
y_m = R\int_{m/R}^{(m+1)/R} y(t) \,{\rm d}t,
\quad\text{$m = 0, 1, \dots, R - 1$.}
$$
This approach is called \term{integrate-and-dump} sampling.  Finally, the samples are quantized to a finite precision.  (In this work, we do not model the final quantization step.)

The fundamental point here is that the sampling rate $R$ is much lower than the Nyquist rate $W$.  We will see that $R$ depends primarily on the number $K$ of significant frequencies that participate in the signal.

\subsection{Implementation and Nonidealities}

Any reasonable system for acquiring continuous-time signals must be
implementable in analog hardware.  The system that we propose is built from robust, readily available components.  This subsection briefly discusses some of the engineering issues.

In practice, we generate the chipping sequence with a pseudorandom number generator.  It is preferable to use pseudorandom numbers for several reasons: they are easier to generate; they are easier to store; and their structure can be exploited by digital algorithms.  Many types of pseudorandom generators can be fashioned from basic hardware components.  For example, the Mersenne twister \cite{MN98:Mersenne-Twister} can be implemented with shift registers.
In some applications, it may suffice just to fix a chipping sequence in advance.

The performance of the random demodulator is unlikely to suffer from the fact that the chipping sequence is not completely random.  We have been able to prove that if the chipping sequence consists of $\ell$-wise independent random variables (for an appropriate value of $\ell$), then the demodulator still offers the same guarantees.  Alon et al.~have demonstrated that shift registers can generate a related class of random variables~\cite{AGHP92:Simple-Constructions}.

The mixer must operate at the Nyquist rate $W$.  Nevertheless, the chipping sequence alternates between the levels $\pm 1$, so the mixer only needs to reverse the polarity of the signal.  It is relatively easy to perform this step using inverters and multiplexers.  Most conventional mixers trade speed for linearity, i.e., fast transitions may result in incorrect products.  Since the random demodulator only needs to reverse polarity of the signal, nonlinearity is not the primary nonideality.  Instead, the bottleneck for the speed of the mixer is the settling times of inverters and multiplexors, which determine the length of time it takes for the output of the mixer to reach steady state.

The sampler can be implemented with an off-the-shelf ADC.  It suffers the same types of nonidealities as any ADC, including thermal noise, aperture jitter, comparator ambiguity, and so forth \cite{Walden99:ADC}.  Since the random demodulator operates at a relatively low sampling rate, we can use high-quality ADCs, which exhibit fewer problems.

In practice, a high-fidelity integrator is not required.  It suffices to perform lowpass filtering before the samples are taken.  It is essential, however, that the impulse response of this filter can be characterized very accurately.

The net effect of these nonidealities is much like the addition of noise to the signal.  The signal reconstruction process is very robust, so it performs well even in the presence of noise.  Nevertheless, we must emphasize that, as with any device that employs mixed signal technologies, an end-to-end random demodulator system must be calibrated so that the digital algorithms are aware of the nonidealities in the output of the analog hardware.

\section{Random Demodulation in Matrix Form} \label{sec:demod-mtx}

In the ideal case, the random demodulator is a linear system that maps a continuous-time signal to a discrete sequence of samples.  To understand its performance, we prefer to express the system in matrix form.  We can then study its properties using tools from matrix analysis and functional analysis.

\subsection{Discrete-Time Representation of Signals}

The first step is to find an appropriate discrete representation for the space of continuous-time input signals.  To that end, note that each $(1/W)$-second block of the signal is multiplied by a random sign.  Then these blocks are aggregated, summed, and sampled.  Therefore, part of the time-averaging performed by the accumulator commutes with the demodulation process.  In other words, we can average the input signal over blocks of duration $1/W$ without affecting subsequent steps.

Fix a time instant of the form $t_n = n / W$ for an integer $n$.  Let $x_n$ denote the average value of the signal $f$ over a time interval of length
$1/W$ starting at $t_n$.  Thus,
\begin{align}
x_n &= \int_{t_n}^{t_n + 1/W} f(t) \idiff{t} \notag \\
&= \sum_{\omega \in \Omega} a_\omega \left[
   \frac{\econst^{-2\pi\iunit\omega / W} - 1}{2\pi\iunit \omega}
   \right] \econst^{-2\pi\iunit \omega t_n} \label{eqn:dt}
\end{align}
with the convention that, for the frequency $\omega = 0$, the
bracketed term equals $1/W$.  Since $\abs{\omega} \leq W/2$,
the bracket never equals zero.  Absorbing the brackets into the
amplitude coefficients, we obtain a discrete-time representation $x_n$ of the signal $f(t)$:
$$
x_n = \sum_{\omega \in \Omega} s_\omega \,
   \econst^{-2\pi\iunit n \omega / W}
\quad\text{for $n = 0, 1, \dots, W - 1$}
$$
where
$$
s_{\omega} = a_{\omega} \left[
   \frac{\econst^{-2\pi\iunit\omega / W} - 1}{2\pi\iunit \omega}
   \right].
$$
In particular, a continuous-time signal that involves only the
frequencies in $\Omega$ can be viewed as a discrete-time signal
comprised of the same frequencies. We refer to the complex vector
$\vct{s}$ as an \term{amplitude vector}, with the understanding that
it contains phase information as well.

The nonzero components of the length-$W$ vector $\vct{s}$ are listed
in the set $\Omega$. We may now express the discrete-time signal
$\vct{x}$ as a matrix--vector product.  Define the $W \times W$
matrix
\begin{multline*}
\mtx{F} = \frac{1}{\sqrt{W}} \begin{bmatrix}
   \econst^{-2\pi\iunit n \omega / W}
\end{bmatrix}_{n, \omega} \qquad\text{where} \\
n = 0, 1, \dots, W - 1
\text{ and } \\
\omega = 0, \pm 1, \dots, \pm \left( \frac{W}{2} - 1\right), \frac{W}{2}.
\end{multline*}
The matrix $\mtx{F}$ is a simply a permuted discrete Fourier transform (DFT) matrix.  In particular, $\mtx{F}$ is unitary and its entries share the magnitude $W^{-1/2}$.

In summary, we can work with a discrete representation
$$
\vct{x} = \mtx{F}\vct{s}
$$
of the input signal.

\subsection{Action of the Demodulator}

We view the random demodulator as a linear system acting on the
discrete form $\vct{x}$ of the continuous-time signal $f$.

First, we consider the effect of random demodulation on the discrete-time signal.  Let $\eps_0, \eps_1, \dots, \eps_{W-1}$ be the chipping sequence.  The demodulation step multiplies each $x_n$, which is the average of $f$ on the $n$th time interval, by the random sign $\eps_n$.  Therefore, demodulation corresponds to the map $\vct{x} \mapsto \mtx{D} \vct{x}$ where
$$
\mtx{D} =
\begin{bmatrix}
\eps_0 \\
& \eps_1 \\
&& \ddots \\
&&& \eps_{W-1}
\end{bmatrix}
$$
is a $W \times W$ diagonal matrix.

Next, we consider the action of the accumulate-and-dump sampler.
Suppose that the sampling rate is $R$, and assume that $R$ divides
$W$.  Then each sample is the sum of $W/R$ consecutive entries of
the demodulated signal.  Therefore, the action of the sampler can be
treated as an $R \times W$ matrix $\mtx{H}$ whose $r$th row has
$W/R$ consecutive unit entries starting in column $rW/R + 1$ for
each $r = 0, 1, \dots, R - 1$.  An example with $R = 3$ and $W = 12$
is
$$
\mtx{H} =
\left[ \begin{array}{cccccccccccc}
1 & 1 & 1 & 1 \\
 &   &   &   & 1 & 1 & 1 & 1 \\
 &   &   &   &   &   &   &   & 1 & 1 & 1 & 1 \\
\end{array} \right].
$$

When $R$ does not divide $W$, two samples may share contributions from a single element of the chipping sequence.  We choose to address this situation by allowing the matrix $\mtx{H}$ to have fractional elements in some of its columns.  An example with $R=3$ and $W=7$ is
$$
\mtx{H} =
\left[ \begin{array}{ccccccc}
1 & 1  & \sqrt{1/3} \\
  &   & \sqrt{2/3} & 1 &  \sqrt{2/3} \\
  &  &   &   &  \sqrt{1/3} & 1 & 1 \\
\end{array} \right].
$$
We have found that this device provides an adequate approximation to
the true action of the system.  For some applications, it may be
necessary to exercise more care.

In summary, the matrix $\mtx{M} = \mtx{HD}$ describes the action of the hardware system on the discrete signal $\vct{x}$.  Each row of the matrix yields a separate sample of the input signal.

The matrix $\Fee = \mtx{MF}$ describes the overall action of the system on the vector $\vct{s}$ of amplitudes. This matrix $\Fee$ has a special place in our analysis, and we refer to it as a \term{random demodulator matrix}.

\subsection{Prewhitening}

It is important to note that the bracket in \eqref{eqn:dt} leads to a nonlinear attenuation of the amplitude coefficients.  In a hardware implementation of the random demodulator, it may be advisable to apply a prewhitening filter to preserve the magnitudes of the amplitude coefficients.  On the other hand, prewhitening amplifies noise in the high frequencies.  We leave this issue for future work.

\section{Signal Recovery Algorithms} \label{sec:algs}

The Shannon sampling theorem provides a simple linear method for reconstructing a bandlimited signal from its time samples.  In contrast, there is no linear process for reconstructing the input signal from the output of the random demodulator because we must incorporate the highly nonlinear sparsity constraint into the reconstruction process.

Suppose that $\vct{s}$ is a sparse amplitude vector, and let $\vct{y} = \Fee\vct{s}$ be the vector of samples acquired by the random demodulator.  Conceptually, the way to recover $\vct{s}$ is to solve the mathematical program
\begin{equation} \label{eqn:l0}
\widehat{\vct{s}} =
\arg\min \pnorm{0}{\vct{v}}
\subjto
\Fee \vct{v} = \vct{y}
\end{equation}
where the $\ell_0$ function $\pnorm{0}{\cdot}$ counts the number of nonzero entries in a vector.  In words, we seek the sparsest amplitude vector that generates the samples we have observed.  The presence of the $\ell_0$ function gives the problem a combinatorial character, which may make it computationally difficult to solve completely.

Instead, we resort to one of the signal recovery algorithms from the sparse approximation or compressive sampling literature.  These techniques fall in two rough classes: convex relaxation and greedy pursuit.  We describe the advantages and disadvantages of each approach in the sequel.

This work concentrates on convex relaxation methods because they are
more amenable to theoretical analysis.  Our purpose here is not to
advocate a specific algorithm but rather to argue that the random
demodulator has genuine potential as a method for acquiring signals
that are spectrally sparse.  Additional research on algorithms will
be necessary to make the technology viable.  See Section~\ref{sec:speed}
for discussion of how quickly we can perform signal recovery with
contemporary computing architectures.

\subsection{Convex Relaxation}
\label{sec:convexrelax}

The problem~\eqref{eqn:l0} is difficult because of the unfavorable properties of the $\ell_0$ function.  A fundamental method for dealing with this challenge is to relax the $\ell_0$ function to the $\ell_1$ norm, which may be viewed as the convex function ``closest'' to $\ell_0$.  Since the $\ell_1$ norm is convex, it can be minimized subject to convex constraints in polynomial time~\cite{CDS99:Atomic-Decomposition}.

Let $\vct{s}$ be the unknown amplitude vector, and let $\vct{y} = \Fee \vct{s}$ be the vector of samples acquired by the random demodulator.  We attempt to identify the amplitude vector $\vct{s}$ by solving the convex optimization problem
\begin{equation} \label{eqn:p1}
\widehat{\vct{s}} = 
\arg\min \pnorm{1}{ \vct{v} } \subjto \Fee
\vct{v} = \vct{y}.
\end{equation}
In words, we search for an amplitude vector that yields the same samples and has the least $\ell_1$ norm.  On account of the geometry of the $\ell_1$ ball, this method promotes sparsity in the estimate $\widehat{\vct{s}}$.

The problem \eqref{eqn:p1} can be recast as a second-order cone program.  In our work, we use an old method, iteratively reweighted least squares (IRLS), for performing the optimization \cite[173ff]{Bjo96:Numerical-Methods}.  It is known that IRLS converges linearly for certain signal recovery problems \cite{DDFG08:Iteratively-Reweighted}.  It is also possible to use interior-point methods, as proposed by Cand{\`e}s et al.~\cite{CRT06:Robust-Uncertainty}.

Convex programming methods for sparse signal recovery problems are very powerful.  Second-order methods, in particular, seem capable of achieving very good reconstructions of signals with wide dynamic range.  Recent work suggests that optimal first-order methods provide similar performance with lower computational overhead~\cite{BBC09:Nesta-Fast}.

\subsection{Greedy Pursuits}
\label{sec:greedy}

To produce sparse approximate solutions to linear systems, we can also use another class of methods based on greedy pursuit.  Roughly, these algorithms build up a sparse solution one step at a time by adding new components that yield the greatest immediate improvement in the approximation error.  An early paper of Gilbert and Tropp analyzed the performance of an algorithm called Orthogonal Matching Pursuit for simple compressive sampling problems~\cite{TG07:Signal-Recovery}.  More recently, work of Needell and Vershynin~\cite{NV07:Uniform-Uncertainty,NV07:ROMP-Stable} and work of Needell and Tropp~\cite{NT08:CoSaMP-Iterative} has resulted in greedy-type algorithms whose theoretical performance guarantees are analogous with those of convex relaxation methods.

In practice, greedy pursuits tend to be effective for problems where the solution is ultra-sparse.  In other situations, convex relaxation is usually more powerful.  On the other hand, greedy techniques have a favorable computational profile, which makes them attractive for large-scale problems.

\subsection{Impact of Noise}

In applications, signals are more likely to be \term{compressible} than to be sparse.  (Compressible signals are not sparse but can be approximated by sparse signals.)  Nonidealities in the hardware system lead to noise in the measurements.  Furthermore, the samples acquired by the random demodulator are quantized.   Modern convex relaxation methods and greedy pursuit methods are robust against all these departures from the ideal model.  We discuss the impact of noise on signal reconstruction algorithms in Section~\ref{sec:rip}.

In fact, the major issue with all compressive sampling methods is not the presence of noise {\em per se}.  Rather, the process of compressing the signal's information into a small number of samples inevitably decreases the signal-to-noise ratio.  In the design of compressive sampling systems, it is essential to address this issue.

\section{Empirical Results} \label{sec:exp}

We continue with an empirical study of the minimum measurement rate required to accurately reconstruct signals with the random demodulator.  Our results are phrased in terms of the Nyquist rate $W$, the sampling rate $R$, and the sparsity level $K$.  It is also common to combine these parameters into scale-free quantities.  The \term{compression factor} $R/W$ measures the improvement in the sampling rate over the Nyquist rate, and the \term{sampling efficiency} $K/R$ measures the number of tones acquired per sample.  Both these numbers range between zero and one.

Our results lead us to an empirical rule for the sampling rate necessary to recover \emph{random} sparse signals using the demodulator system:
\begin{equation} \label{eqn:emp-rate}
R \approx 1.7 K \log(W/K + 1).
\end{equation}
The empirical rate is similar to the weak phase transition threshold that Donoho and Tanner calculated for compressive sensing problems with a Gaussian sampling matrix~\cite{Don::2006::High-Dimensional-Centrally-Symmetric}. The form of this relation also echoes the form of the information bound developed in Section~\ref{subsec:infocontent}.

\subsection{Random Signal Model}
\label{subsec:sigmodel}

We begin with a stochastic model for frequency-sparse discrete signals.  To that end, define the \term{signum function}
$$
\sgn( r \econst^{\iunit \theta} ) \defby \begin{cases}
0, & r = 0, \\
\econst^{\iunit\theta}, & r > 0.
\end{cases}
$$
The model is described in the following box:  

\vspace{1pc}
\begin{center}
\renewcommand{\arraystretch}{1.25}
\begin{tabular}{|rcp{.51\columnwidth}|}
\hline
\multicolumn{3}{|c|}{Model (A) for a random amplitude vector $\vct{s}$} \\
\hline\hline
Frequencies: & $\Omega$ & is a uniformly random set of $K$ frequencies from $\{0, \pm 1, \dots, \pm(W/2 - 1), W/2\}$ \\
\hline
Amplitudes: & $\abs{s_\omega}$ & is \emph{arbitrary} for each $\omega \in \Omega$ \\
   & $s_\omega = 0$ & for each $\omega \notin \Omega$ \\
\hline
Phases: & $\sgn(s_\omega)$ & is i.i.d.\ uniform on the unit circle for each $\omega \in \Omega$ \\
\hline
\end{tabular}
\end{center}
\vspace{1pc}

In our experiments, we set the amplitude of each nonzero coefficient equal to one because the success of $\ell_1$ minimization does not depend on the amplitudes.

\subsection{Experimental Design}

Our experiments are intended to determine the minimum sampling rate $R$ that is necessary to identify a $K$-sparse signal with a bandwidth of $W$ Hz.  In each trial, we perform the following steps:
\begin{enumerate}
\item \textbf{Input Signal.} A \emph{new} amplitude vector $\vct{s}$ is drawn at random according to Model (A) in Section~\ref{subsec:sigmodel}.  The amplitudes all have magnitude one.

\item \textbf{Random demodulator.} A \emph{new} random demodulator $\Fee$ is drawn with parameters $K$, $R$, and $W$.  The elements of the chipping sequence are independent random variables, equally likely to be $\pm 1$.

\item \textbf{Sampling.}  The vector of samples is computed using the expression $\vct{y} = \Fee \vct{s}$.

\item \textbf{Reconstruction.} An estimate $\widehat{\vct{s}}$ of the amplitude vector is computed with IRLS.
\end{enumerate}

For each triple $(K, R, W)$, we perform 500 trials.  We declare the experiment a success when $\vct{s} = \widehat{\vct{s}}$ to machine precision, and we report the smallest value of $R$ for which the empirical failure probability is less than 1\%.

\subsection{Performance Results} \label{sec:synthetic}

We begin by evaluating the relationship between the signal bandwidth
$W$ and the sampling rate $R$ required to achieve a high probability
of successful reconstruction.  Figure~\ref{fig:growW} shows the
experimental results for a fixed sparsity of $K=5$ as $W$ increases
from 128 to 2048 Hz, denoted by solid discs. The solid line shows the
result of a linear regression on the experimental data, namely
$$
R = 1.69K\log(W/K + 1) + 4.51.
$$
The variation
about the regression line is probably due to arithmetic effects that
occur when the sampling rate does not divide the bandlimit.

\begin{center}
\begin{figure}
\centerline{
\includegraphics[width=.9\columnwidth]{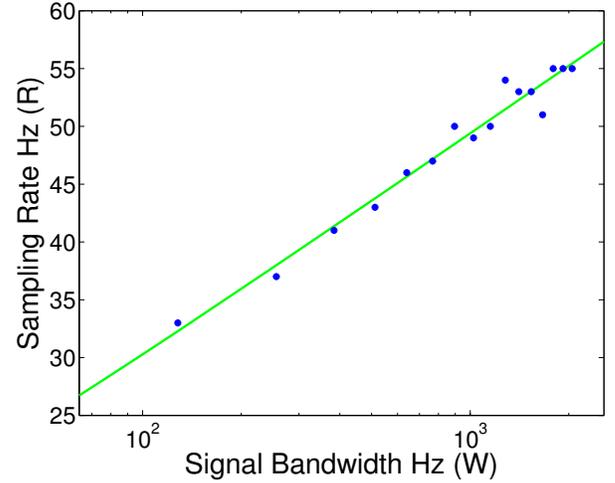}}
\caption{{\sl Sampling rate as a function of signal bandwidth.}  The sparsity is fixed at $K = 5$.  The solid discs mark the lowest sampling rate $R$ that achieves successful reconstruction with probability $0.99$.  The solid line denotes the linear least-squares fit $R = 1.69K\log(W/K + 1) + 4.51$ for the data.}
\label{fig:growW}
\end{figure}
\end{center}

Next, we evaluate the relationship between the sparsity $K$ and the sampling rate $R$. Figure~\ref{fig:growK} shows the experimental results for a fixed chipping rate of $W=512$ Hz as the sparsity $K$ increases from 1 to 64. The figure also shows that the linear regression give a close fit for the data.
$$
R = 1.71K\log(W/K + 1) + 1.00
$$
is the empirical trend.

\begin{center}
\begin{figure}
\centerline{
\includegraphics[width=.9\columnwidth]{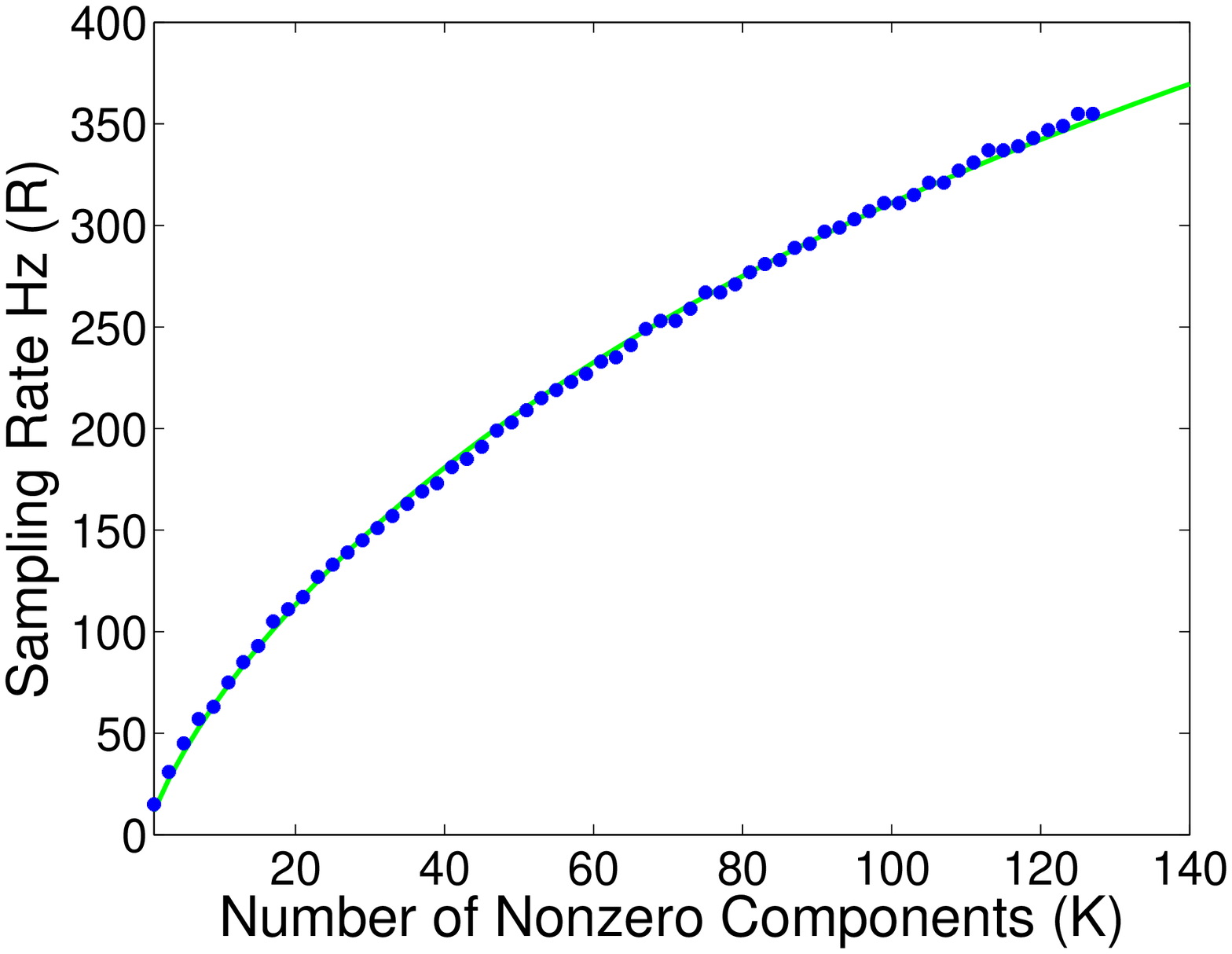}}
\caption{{\sl Samping rate as a function of sparsity.}  The bandlimit is fixed at $W = 512$ Hz.  The solid discs mark the lowest sampling rate that achieves successful reconstruction with probability $0.99$.  The solid line denotes the linear least-squares fit $R = 1.71K\log(W/K + 1) + 1.00$ for the data.}
\label{fig:growK}
\end{figure}
\end{center}

The regression lines from these two experiments suggest that successful reconstruction of signals from Model (A) occurs with high probability when the sampling rate obeys the bound
\begin{equation} \label{eqn:emp-samp-rate}
R \geq 1.7 K\log(W/K + 1).
\end{equation}
Thus, for a fixed sparsity, the sampling rate grows only logarithmically as the Nyquist rate increases. We also note that this bound is similar to those obtained for other measurement schemes that require fully random, dense matrices~\cite{BDDW06:Johnson-Lindenstrauss}.

Finally, we study the threshold that denotes a change from high to low probability of successful reconstruction. This type of {\em phase transition} is a common phenomenon in compressive sampling.  For this experiment, the chipping rate is fixed at $W = 512$ Hz, while the sparsity $K$ and sampling $R$ rate vary.  We record the probability of success as the compression factor $R/W$ and the sampling efficiency $K/R$ vary.

The experimental results appear in Figure~\ref{fig:infocomp}.  Each pixel in this image represents the probability of success for the corresponding combination of system parameters.  Lighter pixels denote higher probability. The dashed line,
$$
\frac{K}{R} = \frac{0.68}{\log(W/K + 1)},
$$
describes the relationship among the parameters where the probability of success drops below 99\%.  The numerical parameter was derived from a linear regression without intercept.

For reference, we compare the random demodulator with a benchmark system that obtains measurements of the amplitude vector $\vct{s}$ by applying a matrix $\Fee$ whose entries are drawn independently from the standard Gaussian distribution.   As the dimensions of the matrix grow, $\ell_1$ minimization methods exhibit a sharp phase transition from success to failure.  The solid line in Figure~\ref{fig:infocomp} marks the location of the precipice, which can be computed analytically with methods from \cite{Don::2006::High-Dimensional-Centrally-Symmetric}.

\begin{center}
\begin{figure}
\centerline{
\includegraphics[width=.9\columnwidth]{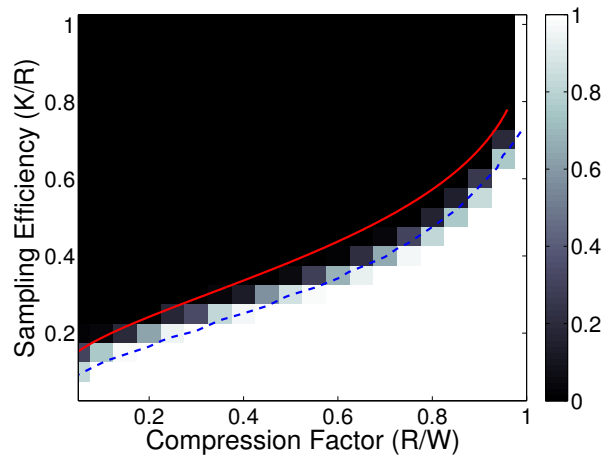}}
\caption{{\sl Probability of success as a function of sampling efficiency and compression factor.}  The shade of each pixel indicates the probability of successful reconstruction for the corresponding combination of parameter values.  The solid line marks the theoretical threshold for a Gaussian matrix; the dashed line traces the 99\% success isocline.}
\label{fig:infocomp}
\end{figure}
\end{center}

\subsection{Example: Analog Demodulation}
\label{sec:example}

We constructed these empirical performance trends using signals drawn from a synthetic model.  To narrow the gap between theory and practice, we performed another experiment to demonstrate that the random demodulator can successfully recover a simple communication signal from samples taken below the Nyquist rate.

In the amplitude modulation (AM) encoding scheme, the transmitted
signal $f_{\rm AM}(t) $ takes the form
\begin{equation}
f_{\rm AM}(t) = A\cos (2\pi \omega_{c} t)\cdot \left(m(t)+C\right),
\end{equation}
where $m(t)$ is the original message signal, $\omega_{c}$ is the carrier frequency, and $A, C$ are fixed values. When the original message signal $m(t)$ has $K$ nonzero Fourier coefficients, then the cosine-modulated signal has only $2K + 2$ nonzero Fourier coefficients.

We consider an AM signal that encodes the message appearing in Figure~\ref{fig:realsig}(a).  The signal was transmitted from a communications device using carrier frequency $\omega_c = 8.2$ KHz, and the received signal was sampled by an ADC at a rate of 32 KHz.  Both the transmitter and receiver were isolated in a lab to produce a clean signal; however, noise is still present on the sampled data due to hardware effects and environmental conditions.

We feed the received signal into a \emph{simulated} random demodulator, where we take the Nyquist rate $W = 32$ KHz and we attempt a variety of sampling rates $R$.  We use IRLS to reconstruct the signal from the random samples, and we perform AM demodulation on the recovered signal to reconstruct the original message. Figures~\ref{fig:realsig}(b)--(d) display reconstructions for a random demodulator with sampling rates $R = 16$ KHz, $8$ KHz, and $3.2$ KHz, respectively.  To evaluate the quality of the reconstruction, we measure the signal-to-noise ratio (SNR)
between the message obtained from the received signal $f_{\rm AM}$ and the message obtained from the output of the random demodulator. The reconstructions achieve SNRs of 27.8 dB, 22.3 dB, and 20.9 dB, respectively.

These results demonstrate that the performance of the random demodulator degrades gracefully as the SNR decreases.  Note, however, that the system will not function at all unless the sampling rate is sufficiently high that we stand below the phase transition.

\begin{figure*}
\begin{center}
\begin{tabular}{cc}
\includegraphics[width=.9\columnwidth]{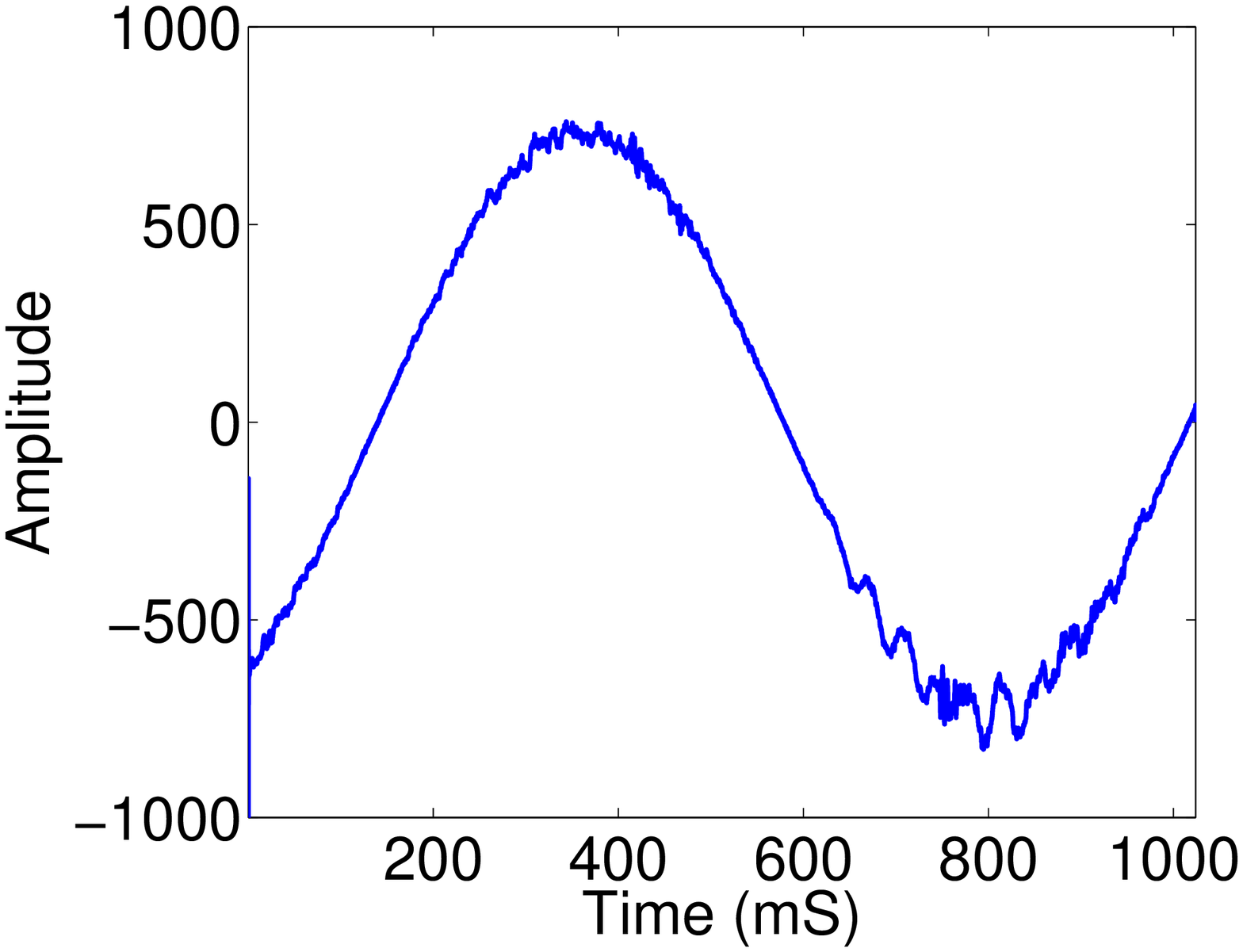}&
\includegraphics[width=.9\columnwidth]{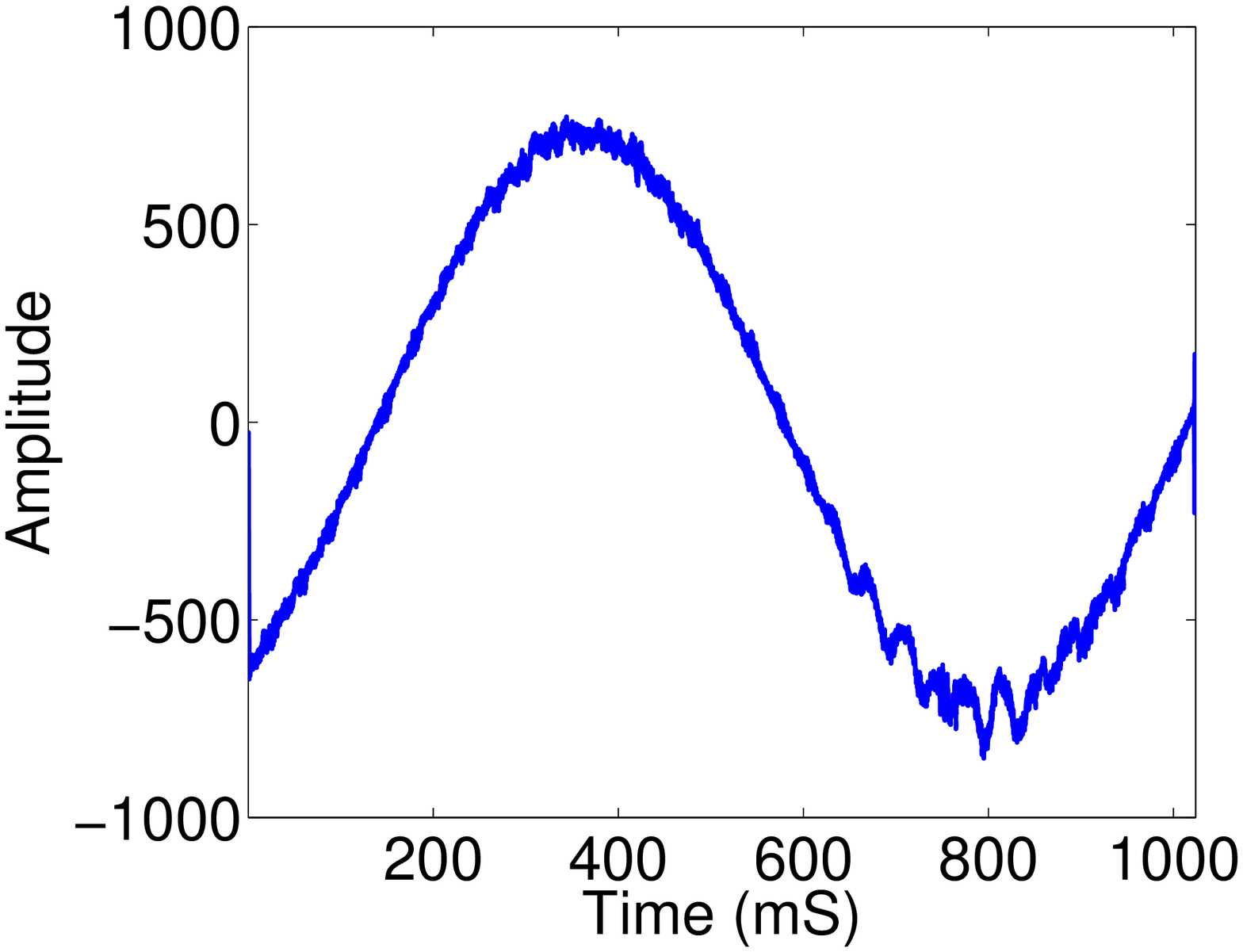}\\
(a)&(b)\\
\includegraphics[width=.9\columnwidth]{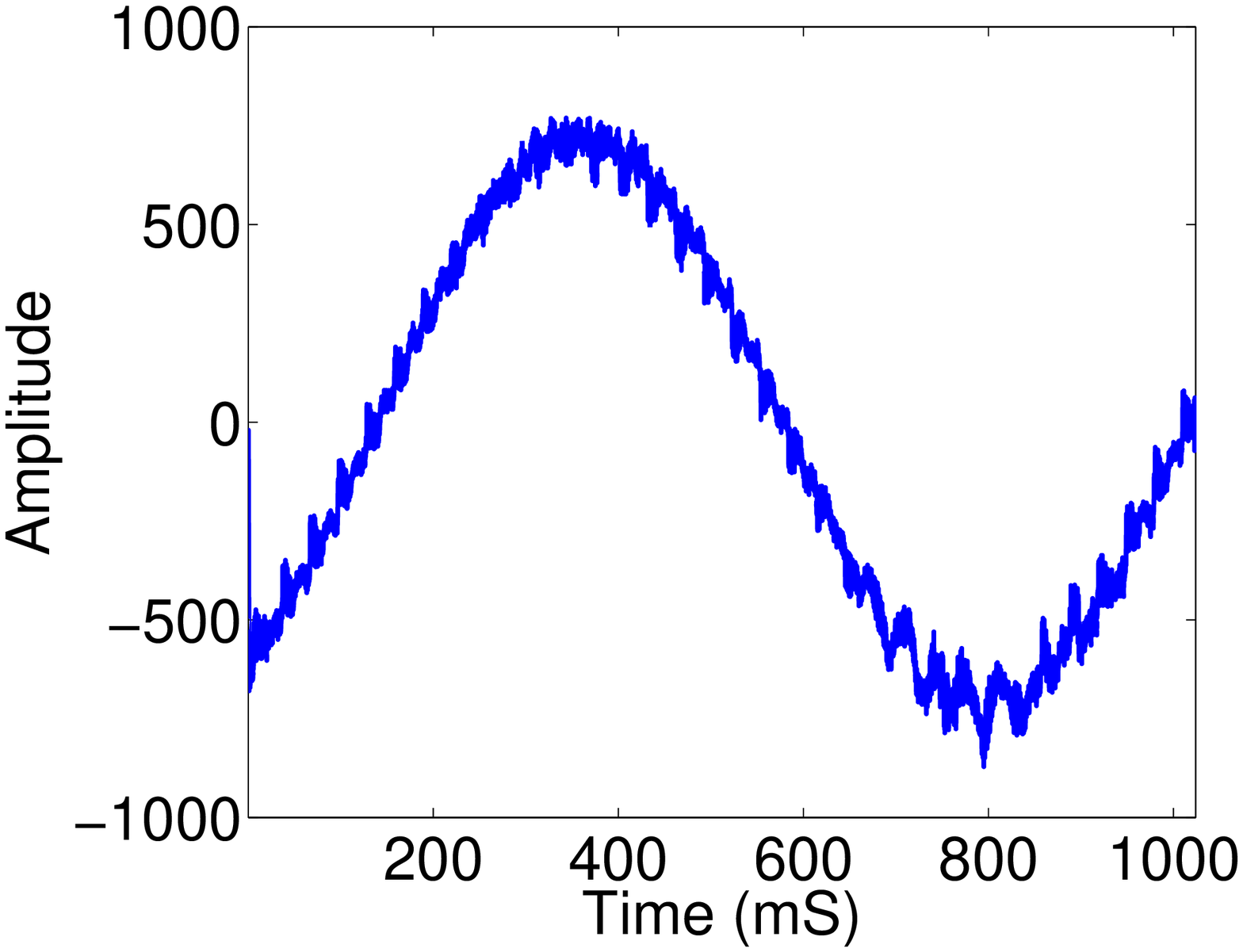}&
\includegraphics[width=.9\columnwidth]{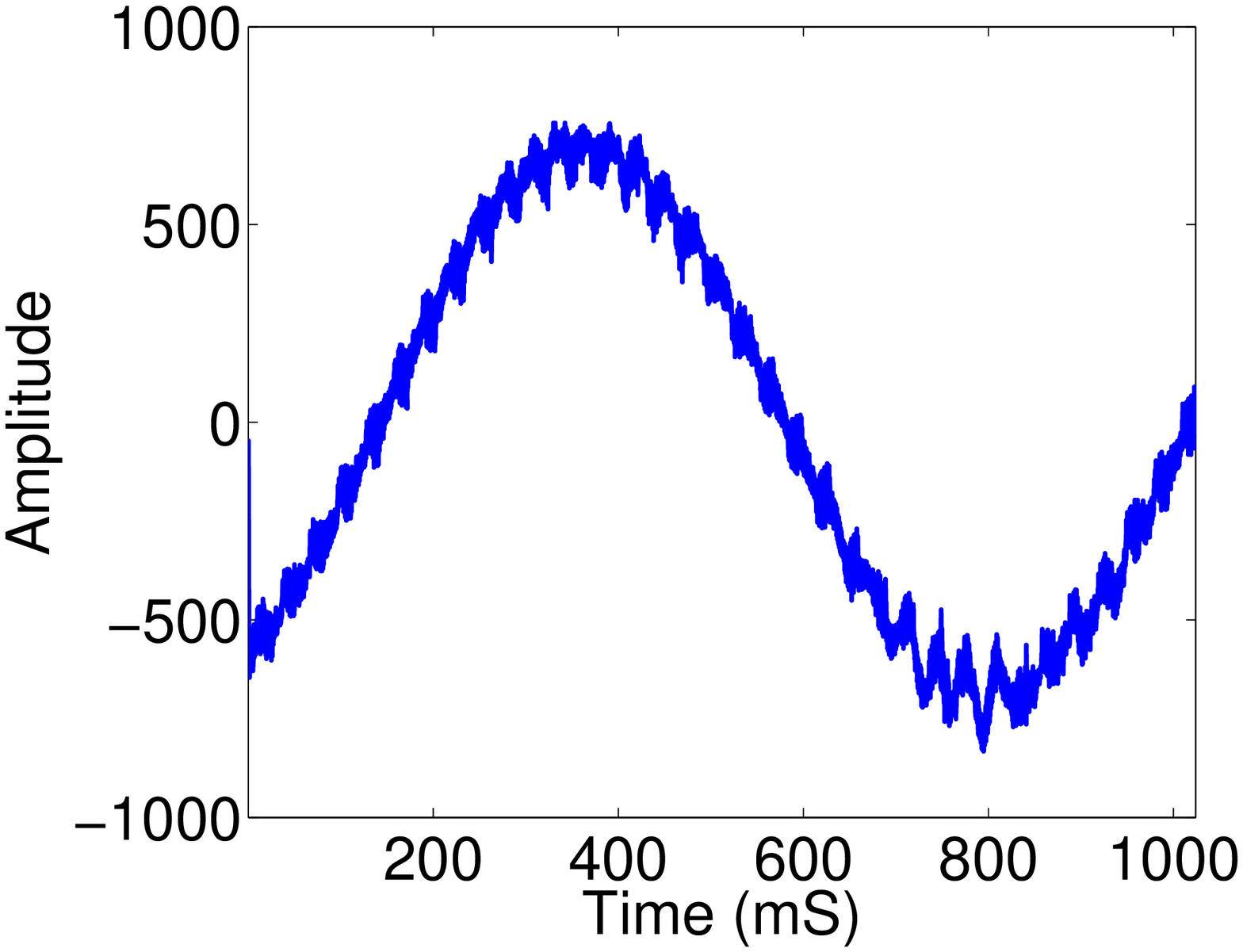}\\
(c)&(d)\\
\end{tabular}
\end{center}
\caption{{\sl Simulated acquisition and reconstruction of an AM signal.}  (a) Message reconstructed from signal sampled at Nyquist rate $W=32$ KHz. (b) Message reconstructed from the output of a simulated random demodulator running at sampling rate $R = 16$ KHz (SNR = 27.8 dB). (c) Reconstruction with $R = 8$ KHz (SNR = 22.3 dB).  (d) Reconstruction with $R =
3.2$ KHz (SNR = 20.9 dB).}
\label{fig:realsig}
\end{figure*}

\section{Theoretical Recovery Results} \label{sec:theory}

We have been able to establish several theoretical guarantees on the performance of the random demodulator system.  First, we focus on the setting described in the numerical experiments, and we develop estimates for the sampling rate required to recover the synthetic sparse signals featured in Section~\ref{sec:synthetic}. Qualitatively, these bounds almost match the system performance we witnessed in our experiments.  The second set of results addresses the behavior of the system for a much wider class of input signals.  This theory establishes that, when the sampling rate is slightly higher, the signal recovery process will succeed---even if the spectrum is not perfectly sparse and the samples collected by the random demodulator are contaminated with noise.

\subsection{Recovery of Random Signals}

First, we study when $\ell_1$ minimization can reconstruct random, frequency-sparse signals that have been sampled with the random demodulator.  The setting of the following theorem precisely matches our numerical experiments.

\begin{thm}[Recovery of Random Signals] \label{thm:rpm}
Suppose that the sampling rate
$$
R \geq \cnst{C} \left[ K \log W + \log^3 W \right]
$$
and that $R$ divides $W$.  The number $\cnst{C}$ is a positive, universal constant.

Let $\vct{s}$ be a random amplitude vector drawn according to Model (A).  Draw an $R \times W$ random demodulator matrix $\Fee$, and let $\vct{y} = \Fee \vct{s}$ be the samples collected by the random demodulator.

Then the solution $\widehat{\vct{s}}$ to the convex program \eqref{eqn:p1} equals $\vct{s}$, except with probability $\bigO(W^{-1})$.
\end{thm}

We have framed the unimportant technical assumption that $R$ divides $W$ to simplify the lengthy argument. See Appendix~\ref{app:random-phase} for the proof.

Our analysis demonstrates that the sampling rate $R$ scales linearly with the sparsity level $K$, while it is logarithmic in the bandlimit $W$.  In other words, the theorem supports our empirical rule~\eqref{eqn:emp-rate} for the sampling rate.  Unfortunately, the analysis does not lead to reasonable estimates for the leading constant.  Still, we believe that this result offers an attractive theoretical justification for our claims about the sampling efficiency of the random demodulator.

The theorem also suggests that there is a small startup cost.  That
is, a minimal number of measurements is required before the
demodulator system is effective.  Our experiments were not refined
enough to detect whether this startup cost actually exists in
practice.

\subsection{Uniformity and Stability}
\label{sec:rip}

Although the results described in the last section provide a satisfactory explanation of the numerical experiments, they are less compelling as a vision of signal recovery in the real world.  Theorem~\ref{thm:rpm} has three major shortcomings.  First, it is unreasonable to assume that tones are located at random positions in the spectrum, so Model (A) is somewhat artificial.  Second, typical signals are not spectrally sparse because they contain background noise and, perhaps, irrelevant low-power frequencies.  Finally, the hardware system suffers from nonidealities, so it only approximates the linear transformation described by the matrix $\mtx{M}$.  As a consequence, the samples are contaminated with noise.

To address these problems, we need an algorithm that provides \term{uniform} and \term{stable} signal recovery. Uniformity means that the approach works for all signals, irrespective of the frequencies and phases that participate.  Stability means that the performance degrades gracefully when the signal's spectrum is not perfectly sparse and the samples are noisy.
Fortunately, the compressive sampling community has developed several powerful signal recovery techniques that enjoy both these properties.

Conceptually, the simplest approach is to modify the convex program \eqref{eqn:p1} to account for the contamination in the samples~\cite{CRT06:Stable-Signal}.
Suppose that $\vct{s} \in \Cspace{W}$ is an arbitrary amplitude
vector, and let $\vct{\nu} \in \Cspace{R}$ be an arbitrary noise
vector that is known to satisfy the bound $\enorm{\vct{\nu}} \leq
\eta$.  Assume that we have acquired the dirty samples
$$
\vct{y} = \Fee \vct{s} + \vct{\nu}.
$$
To produce an approximation of $\vct{s}$, it is natural to solve the noise-aware optimization problem
\begin{equation}\label{eqn:p1-err}
\widehat{\vct{s}} = \argmin \pnorm{1}{\vct{v}} \subjto
\enorm{\Fee\vct{v} - \vct{y}} \leq \eta.
\end{equation}
As before, this problem can be formulated as a second-order cone program, which can be optimized using various algorithms~\cite{CRT06:Stable-Signal}.

We have established a theorem that describes how the optimization-based recovery technique
performs when the samples are acquired using the random demodulator system.  An analogous result holds for the {\sf CoSaMP}
algorithm, a greedy pursuit method with superb guarantees on its runtime~\cite{NT08:CoSaMP-Iterative}.  

\begin{thm}[Recovery of General Signals] \label{thm:stable-recovery}
Suppose that the sampling rate
$$
R \geq \cnst{C} K \log^6 W
$$
and that $R$ divides $W$.  Draw an $R \times W$ random demodulator matrix $\Fee$.  Then  the following statement holds, except with probability $\bigO(W^{-1})$.

Suppose that $\vct{s}$ is an arbitrary amplitude vector and $\vct{\nu}$ is a noise vector with $\enorm{\vct{\nu}} \leq \eta$.  Let $\vct{y} = \Fee \vct{s} + \vct{\nu}$ be the noisy samples collected by the random demodulator.

Then every solution $\widehat{\vct{s}}$ to the convex program \eqref{eqn:p1-err} approximates the target vector $\vct{s}$:
\begin{equation} \label{eqn:21-err-bd}
\enorm{\widehat{\vct{s}} - \vct{s}} \leq \cnst{C} \max\left\{
   \eta, \frac{1}{\sqrt{K}} \pnorm{1}{ \vct{s} - \vct{s}_K } \right\},
\end{equation}
where $\vct{s}_K$ is a best $K$-sparse approximation to $\vct{s}$ with respect to the $\ell_1$ norm.
\end{thm}

The proof relies on the \term{restricted isometry property} (RIP) of Cand{\`e}s--Tao~\cite{CT06:Near-Optimal}.  To establish that the random demodulator verifies the RIP, we adapt ideas of Rudelson--Vershynin~\cite{RV06:Sparse-Reconstruction}.  Turn to Appendix~\ref{app:rip} for the argument.

Theorem~\ref{thm:stable-recovery} is not as easy to grasp as Theorem~\ref{thm:rpm}, so we must spend some time to unpack its meaning.  First, observe that the sampling rate has increased by several logarithmic factors. Some of these factors are probably parasitic, a consequence of the techniques used to prove the theorem.  It seems plausible that, in practice, the actual requirement on the sampling rate is closer to
$$
R \geq \cnst{C} K \log^2 W.
$$
This conjecture is beyond the power of current techniques.

The earlier result, Theorem~\ref{thm:rpm}, suggests that we should draw a new random demodulator each time we want to acquire a signal.  In contrast, Theorem~\ref{thm:stable-recovery} shows that, with high probability, a particular instance of the random demodulator acquires enough information to reconstruct any amplitude vector whatsoever.  This aspect of Theorem~\ref{thm:stable-recovery} has the practical consequence that a random chipping sequence can be chosen in advance and fixed for all time.

Theorem~\ref{thm:stable-recovery} does not place any specific requirements on the amplitude vector, nor does it model the noise contaminating the signal.  Indeed, the approximation error bound \eqref{eqn:21-err-bd} holds generally.  That said, the strength of the error bound depends substantially on the structure of the amplitude vector.  When $\vct{s}$ happens to be $K$-sparse, then the second term in the maximum vanishes.  So the convex programming method still has the ability to recover sparse signals perfectly.

When $\vct{s}$ is not sparse, we may interpret the bound~\eqref{eqn:21-err-bd} as saying that the computed approximation is comparable with the best $K$-sparse approximation of the amplitude vector.  When the amplitude vector has a good sparse approximation, then the recovery succeeds well.  When the amplitude vector does not have a good approximation, then signal recovery may fail completely.  For this class of signal, the random demodulator is not an appropriate technology.

Initially, it may seem difficult to reconcile the two norms that appear in the error bound~\eqref{eqn:21-err-bd}.  In fact, this type of mixed-norm bound is structurally optimal~\cite{CDD06:Compressed-Sensing}, so we cannot hope to improve the $\ell_1$ norm to an $\ell_2$ norm.  Nevertheless, for an important class of signals, the scaled $\ell_1$ norm of the tail is essentially equivalent to the $\ell_2$ norm of the tail.

Let $p \in (0, 1)$.  We say that a vector $\vct{s}$ is \term{$p$-compressible} when its sorted components decay sufficiently fast:
\begin{equation}
\label{eq:compressible}
\abs{s}_{(k)} \leq k^{-1/p}
\quad\text{for}\quad k = 1, 2, 3, \dots.
\end{equation}
Harmonic analysis shows that many natural signal classes are compressible~\cite{DVDD98:Data-Compression}. Moreover, the windowing scheme of Section~\ref{sec:prewindowing} results in compressible signals.

The critical fact is that compressible signals are well approximated by sparse signals.  Indeed, it is straightforward to check that
\begin{align*}
\pnorm{1}{ \vct{s} - \vct{s}_K} &\leq \frac{1}{1/p - 1} \cdot K^{1 - 1/p} \\
\enorm{\vct{s} - \vct{s}_K} &\leq \frac{1}{\sqrt{2/p - 1}} \cdot K^{1/2 - 1/p}.
\end{align*}
Note that the constants on the right-hand side depend only on the level $p$ of compressibility.

For a $p$-compressible signal, the error bound~\eqref{eqn:21-err-bd} reads
$$
\enorm{ \widehat{\vct{s}} - \vct{s} } \leq \cnst{C} \max\left\{
   \eta, K^{1/2 - 1/p} \right\}.
$$
We see that the right-hand side is comparable with the $\ell_2$ norm of the tail of the signal.  This quantitive conclusion reinforces the intuition that the random demodulator is efficient when the amplitude vector is well approximated by a sparse vector.

\subsection{Extension to Bounded Orthobases}

We have shown that the random demodulator is effective at acquiring signals that are spectrally sparse or compressible.  In fact, the system can acquire a much more general class of signals.  The proofs of Theorem~\ref{thm:rpm} and Theorem~\ref{thm:stable-recovery} indicate that the crucial feature of the sparsity basis $\mtx{F}$ is its \term{incoherence} with the Dirac basis.  In other words, we exploit the fact that the entries of the matrix $\mtx{F}$ have small magnitude.  This insight allows us to extend the recovery theory to other sparsity bases with the same property.  We avoid a detailed exposition.  See~\cite{CR07:Sparsity-Incoherence} for more discussion of this type of result.

\section{Windowing} \label{sec:prewindowing}

We have shown that the random demodulator can acquire periodic multitone signals.  The effectiveness of the system for a general signal $f$ depends on how closely we can approximate $f$ on the time interval $[0,1)$ by a periodic multitone signal.  In this section, we argue that many common types of nonperiodic signals can be approximated well using \term{windowing} techniques.  In particular, this approach allows us to capture nonharmonic sinusoids and multiband signals with small total bandwidth.
We offer a brief overview of the ideas here, deferring a detailed analysis for a later publication.

\subsection{Nonharmonic Tones}

First, there is the question of capturing and reconstructing {\em nonharmonic} sinusoids, signals of the form
$$
f(t) = a_{\omega'}\econst^{-2\pi\iunit\omega' t / W}, \quad \omega' \notin \mathbb{Z}.
$$
It is notoriously hard to approximate $f$ on $[0,1)$ using harmonic
sinusoids, because the coefficients $a_\omega$ in
\eqref{eqn:cont-time} are essentially samples of a sinc function.
Indeed, the signal $f_K$, defined as the best approximation of $f$
using $K$ harmonic frequencies, satisfies only
\begin{equation}
\label{eq:sincapprox}
\pnorm{L_2[0,1)}{ f - f_K } \lesssim K^{-1/2},
\end{equation}
a painfully slow rate of decay.

There is a classical remedy to this problem.  Instead of acquiring $f$ directly, we acquire a smoothly windowed version of $f$.  To fix ideas, suppose that $\psi$ is a window function that vanishes outside the interval $[0,1)$.  Instead of measuring $f$, we measure $g = \psi \cdot f$.  The goal is now to reconstruct this {\em windowed} signal $g$.

As before, our ability to reconstruct the windowed signal depends on how well we can approximate it using a periodic multitone signal.  In fact, the approximation rate depends on the smoothness of the window.  When the continuous-time Fourier transform $\widehat{\psi}$ decays like $\omega^{-r}$ for some $r \geq 1$, we have that $g_K$, the best approximation of $g$ using $K$ harmonic frequencies, satisfies
$$
\pnorm{L_2[0,1)}{ g - g_K } \lesssim K^{-r+1/2},
$$
which decays to zero much faster than the error bound~\eqref{eq:sincapprox}.  Thus, to achieve a tolerance of $\eps$, we need only about $\eps^{-1/(r-1/2)}$ terms.

In summary, windowing the input signal allows us to closely
approximate a nonharmonic sinusoid by a periodic multitone signal;
$g$ will be compressible as in \eqref{eq:compressible}.  If there
are multiple nonharmonic sinusoids present, then the number of
harmonic tones required to approximate the signal to a certain
tolerance scales linearly with the number of nonharmonic sinusoids.

\subsection{Multiband Signals}

Windowing also enables us to treat signals that occupy a small band
in frequency.  Suppose that $f$ is a signal whose continuous-time
Fourier transform $\widehat{f}$ vanishes outside an interval of
length $B$.  The Fourier transform of the windowed signal $g = \psi
\cdot f$ can be broken into two pieces.  One part is nonzero only on
the support of $\widehat{f}$; the other part decays like
$\omega^{-r}$ away from this interval.  As a result, we can
approximate $g$ to a tolerance of $\eps$ using a multitone signal
with about $B + \eps^{-1/(r-1/2)}$ terms.

For a multiband signal, the same argument applies to each band.  Therefore, the number of tones required for the approximation scales linearly with the total bandwidth.

\subsection{A Vista on Windows}

To reconstruct the original signal $f$ over a long period of time, we must multiply the signal with overlapping shifts of a window $\psi$.  The window needs to have additional properties for this scheme to work.  Suppose that the set of half-integer shifts of $\psi$ form a partition of unity, i.e.,
$$
\sum\nolimits_{k} \psi(t-k/2) = 1,
\quad t \in \Rspace{}.
$$
After measuring and reconstructing $g_k(t) = \psi(t-k/2) \cdot f(t)$ on each subinterval, we simply add the reconstructions together to obtain a complete reconstruction of $f$.

This windowing strategy relies on our ability to measure the windowed signal $\psi \cdot f$.  To accomplish this, we require a slightly more complicated system architecture.  To perform the windowing, we use an amplifier with a time-varying gain.  Since subsequent windows overlap, we will also need to measure $g_k$ and $g_{k+1}$ simultaneously, which creates the need for two random demodulator channels running in parallel.  In principle, none of these modifications diminishes the practicality of the system.

\section{Discussion}
\label{sec:discussion}

This section develops some additional themes that arise from our work.  First, we discuss the SNR performance and SWAP profile of the random demodulator system.  Afterward, we present a rough estimate of the time required for signal recovery using contemporary computing architectures.  We conclude with some speculations about the technological impact of the device.


\subsection{SNR Performance}

A significant benefit of the random demodulator is that we can control the SNR performance of the reconstruction by optimizing the sampling rate of the back-end ADC, as demonstrated in Section~\ref{sec:example}.  When input signals are spectrally sparse, the random demodulator system can outperform a standard ADC by a substantial margin.

To quantify the SNR performance of an ADC, we consider a standard metric, the effective number of bits (ENOB), which is roughly the actual number of quantization levels possible at a given SNR.  The ENOB is calculated as
\begin{equation}
\hbox{ENOB} = (\hbox{SNR}-1.76)/6.02,
\end{equation}
where the SNR is expressed in dB.   We estimate the ENOB of a modern ADC using the information in Walden's 1999 survey~\cite{Walden99:ADC}.

When the input signal has sparsity $K$ and bandlimit $W$, we can acquire the signal using a random demodulator running at rate $R$, which we estimate using the relation~\eqref{eqn:emp-samp-rate}.  As noted, this rate is typically much lower than the Nyquist rate.  Since low-rate ADCs exhibit much better SNR performance than high-rate ADCs, the demodulator can produce a higher ENOB.

Figure~\ref{fig:snrwalden} charts how much the random demodulator improves on state-of-the-art ADCs when the input signals are assumed to be sparse.  The first panel, Figure~\ref{fig:snrwalden}(a), compares the random demodulator with a standard ADC for signals with the sparsity $K = 10^6$ as the Nyquist rate $W$ varies up to $10^{12}$ Hz.  In this setting, the clear advantages of the random demodulator can be seen in the slow decay of the curve.  The second panel, Figure~\ref{fig:snrwalden}(b), estimates the ENOB as a function of the sparsity $K$ when the bandlimit is fixed at $W = 10^8$ Hz.  Again, a distinct advantage is visible.  But we stress that this improvement is contingent on the sparsity of the signal.

\begin{center}
\begin{figure}
\begin{tabular}{c}
\includegraphics[width=.9\columnwidth]{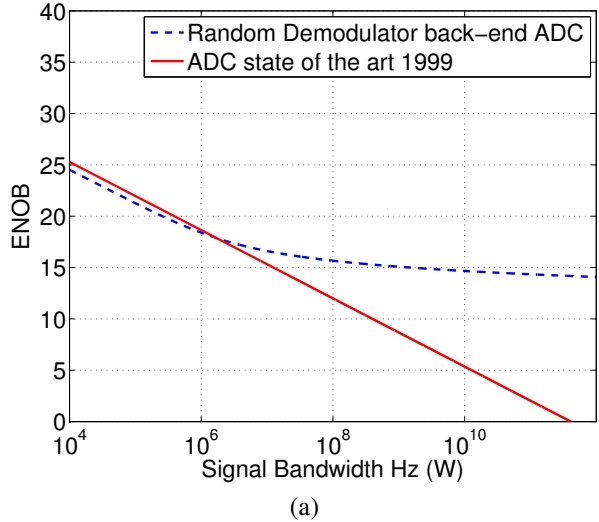} \\ (a) \\
\includegraphics[width=.9\columnwidth]{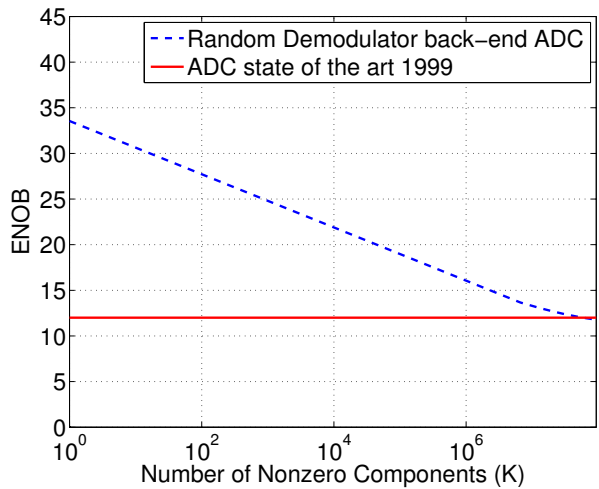}\\ (b)
\end{tabular}
\caption{{\sl ENOB for random demodulator versus a standard ADC.}   The solid lines represent state-of-the-art ADC performance in 1999.  The dashed lines represent the random demodulator performance using an ADC running at the sampling rate $R$ suggested by our empirical work. (a) Performance as a function of the bandlimit $W$ with fixed sparsity $K= 10^{6}$. (b) Performance as a function of the sparsity $K$ with fixed bandlimit $W= 10^{8}$.  Note that the performance of a standard ADC does not depend on $K$.}
\label{fig:snrwalden}
\end{figure}
\end{center}

Although this discussion suggests that the SNR behavior of the 
random demodulator represents a clear improvement over traditional
ADCs, we have neglected some important factors.  The
estimated performance of the demodulator does not take into account
SNR degradations due to the chain of hardware components upstream of
the sampler.  These degradations are caused by factors such as the
nonlinearity of the multiplier and jitter of the pseudorandom
modulation signal.  Thus, it is critical to choose high-quality
components in the design.  Even under this constraint, the 
random demodulator may have a more attractive feasibility and cost
profile than a high-rate ADC.

\subsection{Power Consumption}

In some applications, the {power consumption} of the signal acquisition
system is of tantamount importance.
One widely used figure of merit for ADCs is the quantity
$$
\frac{2^{\rm ENOB} f_s}{P_{\rm diss}},
$$
where $f_s$ is the sampling rate and $P_{\rm diss}$ is the power
dissipation \cite{LeRonRee::2005::Analog-to-Digital-Converters:}.
We propose a slightly modified version of this figure of
merit to compare compressive ADCs with conventional ones:
$$
{\rm FOM} = \frac{2^{{\rm ENOB}-1} W}{P_{\rm diss}(R)},
$$
where we simply replace the sampling rate with the acquisition
bandlimit $W/2$ and express the power dissipated as a function of the
actual sampling rate.
For the random demodulator,
$$
{\rm FOM} \approx \frac{2^{{\rm ENOB}-1} W}
{P_{\rm diss}(1.7K \log(W/K))}
$$
on account of~\eqref{eqn:emp-samp-rate}.
The random demodulator incurs a penalty
in the effective number of bits for a given signal, but it may require
significantly less power to acquire the signal.  This effect becomes
pronounced as the bandlimit $W$ becomes large, which is precisely
where low-power ADCs start to fade.

\subsection{Computational Resources Needed for Signal Recovery}
\label{sec:speed}

Recovering randomly demodulated signals in real-time seems like a daunting task.  This section offers a back-of-the-envelope calculation that supports our claim that current digital computing technology is nearly adequate to achieve real-time recovery rates for signals with frequency components in the gigahertz range.

The computational cost of a sparse recovery algorithm is dominated by repeated application of the system matrix $\mtx{\Phi}$ and its transpose, in both the case where the algorithm solves a convex optimization problem (Section~\ref{sec:convexrelax}) or performs a greedy pursuit (Section~\ref{sec:greedy}).
Recently developed algorithms~\cite{figueiredo07gr,hale08fi,yin08br,NT08:CoSaMP-Iterative,BBC09:Nesta-Fast}
typically produce good solutions with a few hundred applications of the system matrix.

For the random demodulator, the matrix $\Phi$ is a composition of three operations:
\begin{enumerate}
\item	a length-$W$ discrete Fourier transform, which requires $\bigO(W\log W)$ multiplications and additions (via an FFT),
\item	a pointwise multiplication, which uses $W$ multiplies, and
\item	a calculation of block sums, which involves $W$ additions.
\end{enumerate}
Of these three steps, the FFT is by far the most expensive.  Roughly speaking, it takes several hundred FFTs to recover a sparse signal with Nyquist rate $W$ from measurements made by the random demodulator.

Let us fix some numbers so we can estimate the amount of computational time required.  Suppose that we want the digital back-end to output $2^{30}$ (or, about 1 billion) samples per second.  Assume that we compute the samples in blocks of size $2^{14} = 16,384$ and that we use $200$ FFTs for each block.  Since we need to recover $2^{16}$ blocks per second, we have to perform about 13 million 16K-point FFTs in one second.

This amount of computation is substantial, but it is not entirely unrealistic.  For example, a recent benchmark~\cite{williams07sc} reports that the Cell processor performs a 16K-point FFT at 37.6 Gflops/s, which translates%
\footnote{A single $n$-point FFT requires around $2.5n\log_2 n$ flops.}
to about $1.5\times 10^{-5}$\,s.  The nominal time to recover a single block is around 3\,ms, so the total cost of processing $2^{16}$ blocks is around 200\,s.  Of course, the blocks can be recovered in parallel, so this factor of 200 can be reduced significantly using parallel or multicore architectures.


The random demodulator offers an additional benefit.  The samples collected by the system automatically compress a sparse signal into a minimal amount of data. As a result, the \emph{storage} and \emph{communication} requirements of the signal acquisition system are also reduced.  Furthermore, no extra processing is required after sampling to compress the data, which decreases the complexity of the required hardware and its power consumption.

\subsection{Technological Impact}

The easiest conclusion to draw from the bounds on SNR performance is that
the random demodulator may allow us to acquire high-bandwidth signals that
are not accessible with current technologies.  These applications still
require high-performance analog and digital signal processing technology,
so they may be several years (or more) away.

A more subtle conclusion is that the random demodulator enables
us to perform certain signal processing tasks using devices with
a more favorable size, weight, and power (SWAP) profile.  We believe that
these applications will be easier to achieve in the near future because
suitable ADC and DSP technology is already available.

\section{Related Work}
\label{sec:related}

Finally, we describe connections between the random demodulator and
other approaches to sampling signals that contain limited information.

\subsection{Origins of Random Demodulator}

The random demodulator was introduced in two earlier
papers~\cite{laskaa2i,laska06}, which offer preliminary work on the
system architecture and experimental studies of the system's
behavior.  The current paper can be seen as an expansion of these
articles, because we offer detailed information on performance
trends as a function of the signal bandlimit, the sparsity, and the
number of samples.  We have also developed theoretical
foundations that support the empirical results.  This analysis
was initially presented by the first author at SampTA 2007.

\subsection{Compressive Sampling}

The most direct precedent for this paper is the theory of
compressive sampling.  This field, initiated in the
papers~\cite{CRT06:Robust-Uncertainty,Don06:Compressed-Sensing}, has
shown that random measurements are an efficient and practical method
for acquiring compressible signals.  For the most part, compressive
sampling has concentrated on finite-length, discrete-time signals.
One of the innovations in this work is to transport the
continuous-time signal acquisition problem into a setting where
established techniques apply.  Indeed, our empirical and theoretical
estimates for the sampling rate of the random demodulator echo
established results from the compressive sampling literature
regarding the number of measurements needed to acquire sparse
signals.

\subsection{Comparison with Nonuniform Sampling}

Another line of research \cite{CRT06:Robust-Uncertainty,GST08:Tutorial-Fast} in compressive sampling has shown that frequency-sparse, periodic, bandlimited signals can be acquired by sampling nonuniformly in time at an average rate comparable with \eqref{eqn:Rrpm}.  This type of nonuniform sampling can be implemented by changing the clock input to a standard ADC.  Although the random demodulator system involves more components, it has several advantages over nonuniform sampling.

First, nonuniform samplers are extremely sensitive to {\em timing jitter}.  Consider the problem of acquiring a signal with high-frequency components by means of nonuniform sampling.  Since the signal values change rapidly, a small error in the sampling time can result in an erroneous sample value.  The random demodulator, on the other hand, benefits from the integrator, which effectively lowers the bandwidth of the input into the ADC. Moreover, the random demodulator uses a uniform clock, which is more stable to generate.

Second, the SNR in the measurements from a random demodulator is
much higher than the SNR in the measurements from a nonuniform
sampler. Suppose that we are acquiring a single sinusoid with unit
amplitude. On average, each sample has magnitude $2^{-1/2}$, so if
we take $W$ samples at the Nyquist rate, the total energy in the
samples is $W/2$.  If we take $R$ nonuniform samples the total
energy will be on average $R/2$.  In contrast, if we take $R$
samples with the random demodulator, each sample has an approximate
magnitude of $\sqrt{W/R}$, so the total energy in the samples is
about $W$.  In consequence, signal recovery using samples from the
random demodulator is more robust against additive noise.

\subsection{Relationship with Random Convolution}

As illustrated in Figure~\ref{fig:demod}, the random demodulator can be interpreted in the frequency domain as a convolution of a sparse signal with a random waveform, followed by lowpass filtering.  An idea closely related to this---convolution with a random waveform followed by subsampling---has appeared in the compressed sensing literature \cite{TWDBB06:Random-Filters,romberg09co,haupt08to,rauhut09ci}.  In fact, if we replace the integrator with an ideal lowpass filter, so that we are in essence taking $R$ consecutive samples of the Fourier transform of the demodulated signal, the architecture would be very similar to that analyzed in \cite{romberg09co}, with the roles of time and frequency reversed.  The main difference is that \cite{romberg09co} relies on the samples themselves being randomly chosen rather than consecutive (this extra randomness allows the same sensing architecture to be used with any sparsity basis).

\subsection{Multiband Sampling Theory}

The classical approach to recovering bandlimited signals from time
samples is, of course, the well-known method associated with the
names of Shannon, Nyquist, Whittaker, Kotel'nikov, and others.  In
the 1960s, Landau demonstrated that \emph{stable} reconstruction of
a bandlimited signal demands a sampling rate no less than the
Nyquist rate~\cite{Lan67:Sampling-Data}.  Landau also considered
multiband signals, those bandlimited signals whose spectrum is
supported on a collection of frequency intervals.  Roughly speaking,
he proved that a sequence of time samples cannot stably determine a
multiband signal unless the average sampling rate exceeds
the measure of the occupied part of the spectrum~\cite{Lan67:Necessary-Density}.

In the 1990s, researchers began to consider practical sampling
schemes for acquiring multiband signals.  The earliest effort is
probably a paper of Feng and Bresler, who assumed that the band
locations were known in advance~\cite{FengBresler}.  See also the
subsequent work~\cite{VenkataramaniBresler}.  Afterward, Mishali and
Eldar showed that related methods can be used to acquire multiband
signals \emph{without} knowledge of the band
locations~\cite{MishaliEldar}.

Researchers have also considered generalizations of the multiband
signal model.  These approaches model signals as arising from a
union of subspaces.  See, for example,~\cite{Eldar,EM09:Robust-Recovery}
and their references.

\subsection{Parallel Demodulator System}

Very recently, it has been shown that a bank of random demodulators
can be used for blind acquisition of multiband
signals~\cite{MET08:Efficient-Sampling,ME09:Theory-Practice}.
This system uses different parameter settings from the system
described here, and it processes samples in a rather different
fashion.  This work is more closely connected with multiband
sampling theory than with compressive sampling, so we refer the
reader to the papers for details.

\subsection{Finite Rate of Innovation}

Vetterli and collaborators have developed an alternative framework
for sub-Nyquist sampling and reconstruction, called \emph{finite rate of
innovation} sampling, that passes an analog signal $f$ having $K$
degrees of freedom per second through a linear time-invariant filter
and then samples at a rate above $2K$ Hz. Reconstruction is then
performed by rooting a high-order polynomial
\cite{FROITSP,FROITSP2,FROISPMag,FROISines}. While this sampling
rate is less than the $\bigO( K \log(W / K) )$ required by the
random demodulator, a proof of the numerical stability of this
method remains elusive.

\subsection{Sublinear FFTs}

During the 1990s and the early years of the present decade, a
separate strand of work appeared in the literature on theoretical
computer science.  Researchers developed computationally efficient
algorithms for approximating the Fourier transform of a signal,
given a small number of random time samples from a structured grid.
The earliest work in this area was due to Kushilevitz and
Mansour~\cite{KM93:Learning-Decision}, while the method was
perfected by Gilbert et al.~\cite{spfour,gilbert04im}.
See~\cite{GST08:Tutorial-Fast} for a tutorial on these ideas.
These schemes provide an alternative approach to sub-Nyquist
ADCs~\cite{rsac,ragheb} in certain settings.

\appendices

\section{The Random Demodulator Matrix} \label{app:rand-demod-mtx}

This appendix collects some facts about the random demodulator matrix that we use to prove the recovery bounds.

\subsection{Notation}

Let us begin with some notation.  First, we abbreviate $\oneton{W} = \{ 1, 2, \dots, W \}$.  We write ${}^\adj$ for the complex conjugate transpose of a scalar, vector, or matrix.  The symbol $\norm{\cdot}$ denotes the spectral norm of a matrix, while $\fnorm{\cdot}$ indicates the Frobenius norm.  We write $\pnorm{\max}{\cdot}$ for the maximum absolute entry of a matrix.  Other norms will be introduced as necessary.  The probability of an event is expressed as $\Prob{\cdot}$, and we use $\Expect$ for the expectation operator.  For conditional expectation, we use the notation
$\Expect_X Z$, which represents integration with respect to $X$, holding all other variables fixed.
For a random variable $Z$, we define its $L_p$ norm
$$
\Expect^p(Z) = \left( \Expect \abs{Z}^p \right)^{1/p}.
$$
We sometimes omit the parentheses when there is no possibility of confusion.
Finally, we remind the reader of the analyst's convention that roman letters $\cnst{c}$, $\cnst{C}$, etc.\ denote universal constants that may change at every appearance.

\subsection{Background}

This section contains a potpourri of basic results that will be helpful to us.

We begin with a simple technique for bounding the moments of a maximum.  Consider an arbitrary set $\{ Z_1, \dots, Z_N \}$ of random variables.  It holds that
\begin{equation} \label{eqn:moment-max}
\Expect^p( \max\nolimits_{j} Z_j )
\leq N^{1/p} \max\nolimits_{j} \Expect^p( Z_j ).
\end{equation}
To check this claim, simply note that
\begin{multline*}
\left[ \Expect \max\nolimits_j \abs{Z_j}^p \right]^{1/p}
   \leq \left[ \Expect \sum\nolimits_j \abs{Z_j}^p \right]^{1/p} \\
   = \left[ \sum\nolimits_j \Expect \abs{Z_j}^p \right]^{1/p}
   \leq \left[ N \cdot \max\nolimits_j \Expect \abs{Z_j}^p \right]^{1/p}.
\end{multline*}
In many cases, this inequality yields essentially sharp results for the appropriate choice of $p$.

The simplest probabilistic object is the \term{Rademacher} random variable, which takes the two values $\pm 1$ with equal likelihood.  A sequence of independent Rademacher variables is referred to as a \term{Rademacher sequence}.  A \term{Rademacher series} in a Banach space $X$ is a sum of the form
$$
\sum\nolimits_{j=1}^\infty \xi_j \vct{x}_j
$$
where $\{ \vct{x}_j \}$ is a sequence of points in $X$ and $\{\xi_j\}$ is an (independent) Rademacher sequence.  

For Rademacher series with scalar coefficients, the most important result is the inequality of Khintchine.  The following
sharp version is due to Haagerup~\cite{Haa82:Best-Constants}.

\begin{prop}[Khintchine] \label{prop:khintchine}
Let $p \geq 2$.  For every sequence $\{a_j\}$ of complex scalars,
$$
\Expect^p \abs{ \sum\nolimits_j \xi_j a_j }
   \leq \cnst{C}_p \left[ \sum\nolimits_j \abssq{a_j} \right]^{1/2},
$$
where the optimal constant
$$
\cnst{C}_p = \left[ \frac{p!}{2^{p/2} (p/2)!} \right]^{1/p}
   \leq 2^{1/2p} \econst^{-0.5} \sqrt{p}.
$$
\end{prop}

\vspace{.5pc}

This inequality is typically established only for real scalars, but the real case implies that the complex case holds with the same constant.

Rademacher series appear as a basic tool for studying sums of independent random variables in a Banach space, as illustrated in the following proposition~\cite[Lem.~6.3]{LT91:Probability-Banach}.

\begin{prop}[Symmetrization] \label{prop:symmetrization}
Let $\{Z_j\}$ be a finite sequence of independent, zero-mean random variables taking values in a Banach space $X$.  Then
$$
\Expect^p \pnorm{X}{ \sum\nolimits_j Z_j }
   \leq 2 \Expect^p \pnorm{X}{ \sum\nolimits_j \xi_j Z_j },
$$
where $\{ \xi_j \}$ is a Rademacher sequence independent of $\{Z_j\}$.
\end{prop}

In words, the moments of the sum are controlled by the moments of the associated Rademacher series.  The advantage of this approach is that we can condition on the choice of $\{Z_j\}$ and apply sophisticated methods to estimate the moments of the Rademacher series.

Finally, we need some facts about \term{symmetrized} random variables.  Suppose that $Z$ is a zero-mean random variable that takes values in a Banach space $X$.  We may define the symmetrized variable $Y = Z - Z'$, where $Z'$ is an independent copy of $Z$.  The tail of the symmetrized variable $Y$ is closely related to the tail of $Z$.  Indeed,
\begin{equation} \label{eqn:symm-tail}
\Prob{ \pnorm{X}{Z} > 2 \Expect \pnorm{X}{Z} + u }
   \leq \Prob{ \pnorm{X}{Y} > u }.
\end{equation}
This relation follows from \cite[Eqn.~(6.2)]{LT91:Probability-Banach} and the fact that $\operatorname{Med}(Y) \leq 2 \Expect Y$ for every nonnegative random variable $Y$.

\subsection{The Random Demodulator Matrix}

We continue with a review of the structure of the random demodulator matrix, which is the central object of our affection.  Throughout the appendices, we assume that the sampling rate $R$ divides the bandlimit $W$.  That is,
$$
W / R \in \mathbb{Z}.
$$
Recall that the $R \times W$ random demodulator matrix $\Fee$ is defined via the product
$$
\Fee = \mtx{HDF}.
$$
We index the $R$ rows of the matrix with the Roman letter $r$, while we index the $W$ columns with the Greek letters $\alpha, \omega$.  It is convenient to summarize the properties of the three factors.

The matrix $\mtx{F}$ is a permutation of the $W \times W$ discrete Fourier transform matrix.  In particular, it is unitary, and each of its entries has magnitude $W^{-1/2}$.

The matrix $\mtx{D}$ is a random $W \times W$ diagonal matrix.  The entries of $\mtx{D}$ are the elements $\eps_1, \dots, \eps_W$ of the chipping sequence, which we assume to be a Rademacher sequence.  Since the nonzero entries of $\mtx{D}$ are $\pm 1$, it follows that the matrix is unitary.

The matrix $\mtx{H}$ is an $R \times W$ matrix with 0--1 entries.  Each of its rows contains a block of $W/R$ contiguous ones, and the rows are orthogonal.  These facts imply that
\begin{equation} \label{eqn:H-norm}
\norm{\mtx{H}} = \sqrt{W/R}.
\end{equation}
To keep track of the locations of the ones, the following notation is useful.  We write
$$
j \sim r
\quad\text{when}\quad
(r-1)W / R < j \leq rW/R.
$$
Thus, each value of $r$ is associated with $W/R$ values of $j$.  The entry $h_{rj} = 1$ if and only if $j \sim r$. 

The spectral norm of $\Fee$ is determined completely by its structure.
\begin{equation} \label{eqn:spec-norm}
\norm{\Fee} = \norm{\mtx{HDF}} = \norm{\mtx{H}} = \sqrt{W/R}
\end{equation}
because $\mtx{F}$ and $\mtx{D}$ are unitary.

Next, we present notations for the entries and columns of the random demodulator.  For an index $\omega \in \oneton{W}$, we write $\atom_{\omega}$ for the $\omega$th column of $\Fee$.  Expanding the matrix product, we can express
\begin{equation} \label{eqn:col-expr}
\atom_{\omega} = \sum\nolimits_{j=1}^W \eps_j f_{j\omega} \vct{h}_j
\end{equation}
where $f_{j\omega}$ is the $(j, \omega)$ entry of $\mtx{F}$ and $\vct{h}_j$ is the $j$th column of $\mtx{H}$.
Similarly, we write $\varphi_{r\omega}$ for the $(r, \omega)$ entry of the matrix $\Fee$.  The entry can also be expanded as a sum
\begin{equation} \label{eqn:entry-expr}
\varphi_{r\omega} = \sum\nolimits_{j \sim r} \eps_j f_{j\omega}.
\end{equation}

Finally, given a set $\Omega$ of column indices, we define the column submatrix $\Fee_{\Omega}$ of $\Fee$ whose columns are listed in $\Omega$.

\subsection{A Componentwise Bound}

The first key property of the random demodulator matrix is that its entries are small with high probability.  We apply this bound repeatedly in the subsequent arguments.

\begin{lemma} \label{lem:entry-bd}
Let $\Fee$ be an $R \times W$ random demodulator matrix.  When $2 \leq p \leq 4\log W$, we have
$$
\Expect^p \pnorm{\max}{ \Fee }
   \leq \sqrt{\frac{6\log W}{R}}.
$$
Furthermore,
$$
\Prob{ \pnorm{\max}{\Fee}  > \sqrt{\frac{10\log W}{R}} }
   \leq W^{-1}.
$$
\end{lemma}

\begin{proof}
As noted in \eqref{eqn:entry-expr}, each entry of $\Fee$ is a Rademacher series:
$$
\varphi_{r\omega} = \sum\nolimits_{j \sim r} \eps_j f_{j\omega}.
$$
Observe that
$$
\sum\nolimits_{j \sim r} \abssq{f_{j\omega}}
   = \frac{W}{R} \cdot \frac{1}{W}
   =\frac{1}{R}.
$$
because the entries of $\mtx{F}$ all share the magnitude $W^{-1/2}$.  Khintchine's inequality, Proposition~\ref{prop:khintchine}, provides a bound on the moments of the Rademacher series:
$$
\Expect^{p} \abs{\varphi_{r\omega}}
   \leq \frac{\cnst{C}_{p}}{\sqrt{R}}
$$
where $\cnst{C}_p \leq 2^{0.25} \econst^{-0.5} \sqrt{p}$.

We now compute the moments of the maximum entry.  Let
$$
M = \pnorm{\max}{\Fee} = \max\nolimits_{r,\omega} \abs{ \varphi_{r\omega} }.
$$
Select $q = \max\{p, 4\log W\}$, and invoke H{\"o}lder's inequality.
$$
\Expect^{p} M \leq \Expect^{q} \max\nolimits_{r,\omega} \abs{\varphi_{r\omega}}.
$$
Inequality \eqref{eqn:moment-max} yields
$$
\Expect^{p} M
   \leq (RW)^{1/q} \max\nolimits_{r,\omega} \Expect^{q} \abs{\varphi_{r\omega}}.
$$
Apply the moment bound for the individual entries to reach
$$
\Expect^{p} M \leq \frac{\cnst{C}_{q} (RW)^{1/q}}{\sqrt{R}}.
$$
Since $R \leq W$ and $q \geq 4\log W$, we have $(RW)^{1/q} \leq \econst^{0.5}$.  Simplify the constant to discover that
$$
\Expect^{p} M \leq 2^{1.25} \sqrt{\frac{\log W}{R}}.
$$
A numerical bound yields the first conclusion.

To obtain a probability bound, we apply Markov's inequality, which is the standard device.  Indeed,
$$
\Prob{ M > u }
   = \Prob{ M^q > u^q }
   \leq \left[ \frac{ \Expect^{q} M }{ u } \right]^{q}.
$$
Choose $u = \econst^{0.25} \Expect^{q} M$ to obtain
$$
\Prob{ M > 2^{1.25}\econst^{0.25} \sqrt{\frac{\log W}{R}} }
   \leq \econst^{-\log W}
   = W^{-1}.
$$
Another numerical bound completes the demonstration.
\end{proof}

\subsection{The Gram Matrix}

Next, we develop some information about the inner products between columns of the random demodulator.  Let $\alpha$ and $\omega$ be column indices in $\oneton{W}$.  Using the expression \eqref{eqn:col-expr} for the columns, we find that
$$
\ip{ \atom_\alpha }{ \atom_\omega }
   = \sum\nolimits_{j,k=1}^W \eps_j \eps_k \eta_{jk} f_{j\alpha}^\adj f_{k\omega}
$$
where we have abbreviated $\eta_{jk} = \ip{ \vct{h}_k }{ \vct{h}_j }$.  Since $\eta_{jj} = 1$, the sum of the diagonal terms is
$$
\sum\nolimits_j f_{j\alpha}^\adj f_{j\omega}
   = \begin{cases} 1, & \alpha = \omega \\ 0, & \alpha \neq \omega \end{cases}
$$
because the columns of $\mtx{F}$ are orthonormal.  Therefore,
$$
\ip{ \atom_\alpha }{ \atom_\omega } = \delta_{\alpha\omega}
   + \sum\nolimits_{j\neq k} \eps_j \eps_k \eta_{jk} f_{j\alpha}^\adj f_{k\omega}
$$
where $\delta_{\alpha\omega}$ is the Kronecker delta.

The Gram matrix $\Fee^\adj \Fee$ tabulates these inner products.  As a result, we can express the latter identity as
$$
\Fee^\adj \Fee = \Id + \mtx{X}
$$
where
$$
x_{\alpha\omega} = \sum\nolimits_{j\neq k} \eps_j \eps_k \eta_{jk} f_{j\alpha}^\adj f_{k\omega}.
$$
It is clear that $\Expect \mtx{X} = \mtx{0}$, so that
\begin{equation} \label{eqn:gram-expect}
\Expect \Fee^\adj \Fee = \Id.
\end{equation}
We interpret this relation to mean that, on average, the columns of $\Fee$ form an orthonormal system.  This ideal situation cannot occur since $\Fee$ has more columns than rows.  The matrix $\mtx{X}$ measures the discrepancy between reality and this impossible dream.  Indeed, most of argument can be viewed as a collection of norm bounds on $\mtx{X}$ that quantify different aspects of this deviation.

Our central result provides a uniform bound on the entries of $\mtx{X}$ that holds with high probability.  In the sequel, we exploit this fact to obtain estimates for several other quantities of interest.

\begin{lemma} \label{lem:X-max-bd}
Suppose that $R \geq 2 \log W$.  Then
$$
\Prob{ \pnorm{\max}{ \mtx{X} } > \cnst{C} \sqrt{\frac{\log W}{R}} }
   \leq W^{-1}.
$$
\end{lemma}

\vspace{.5pc}

\begin{proof}
The object of our attention is
$$
M = \pnorm{\max}{ \mtx{X} } = \max_{\alpha, \omega}
\abs{\sum\nolimits_{j\neq k} \eps_j \eps_k \eta_{jk} f_{j\alpha} f_{k\omega}^\adj}.
$$
A random variable of this form is called a \term{second-order
Rademacher chaos}, and there are sophisticated methods available for
studying its distribution.  We use these ideas to develop a bound on
$\Expect^p M$ for $p = 2 \log W$.  Afterward, we apply Markov's
inequality to obtain the tail bound.

For technical reasons, we must rewrite the chaos before we start making estimates.  First, note that we can express the chaos in a more symmetric fashion:
$$
M = \max_{\alpha, \omega}
\abs{\sum\nolimits_{j\neq k} \eps_j \eps_k \cdot \eta_{jk} \cdot \half
(f_{j\alpha} f_{k\omega}^\adj + f_{k\alpha} f_{j\omega}^\adj)}.
$$
It is also simpler to study the real and imaginary parts of the chaos separately.  To that end, define
\begin{align*}
a_{jk}^{\alpha\omega} &=
\begin{cases}
\eta_{jk} \cdot \half \real
(f_{j\alpha} f_{k\omega}^\adj + f_{k\alpha} f_{j\omega}^\adj), \quad & j \neq k \\
0, & j = k \\
\end{cases} \\
b_{jk}^{\alpha\omega} &=
\begin{cases}
\eta_{jk} \cdot \half \imag
(f_{j\alpha} f_{k\omega}^\adj + f_{k\alpha} f_{j\omega}^\adj), \quad & j \neq k \\
0, & j = k \\
\end{cases}
\end{align*}
With this notation,
\begin{align*}
M &\leq
\max_{\alpha, \omega}
\abs{\sum\nolimits_{j, k} \eps_j \eps_k a_{jk}^{\alpha\omega}} +
\max_{\alpha, \omega}
\abs{\sum\nolimits_{j, k} \eps_j \eps_k b_{jk}^{\alpha\omega}} \\
   &\defby M_{\real} + M_{\imag}.
\end{align*}
We focus on the real part since the imaginary part receives an identical treatment.

The next step toward obtaining a tail bound is to decouple the chaos.  Define
$$
Y = \max_{\alpha, \omega}
\abs{\sum\nolimits_{j, k} \eps_j' \eps_k a_{jk}^{\alpha\omega}}
$$
where $\{\eps_j'\}$ is an independent copy of $\{\eps_j\}$.  Standard decoupling results \cite[Prop.~1.9]{BT87:Invertibility-Large} imply that
$$
\Expect^p M_{\real} \leq 4 \Expect^p Y.
$$
So it suffices to study the decoupled variable $Y$.

To approach this problem, we first bound the moments of each random variable that participates in the maximum. Fix a pair $(\alpha, \omega)$ of frequencies.  We omit the superscript $\alpha$ and $\omega$ to simplify notations.  Define the auxiliary random variable
$$
Z = Z_{\alpha\omega} = \abs{\sum\nolimits_{j, k} \eps_j' \eps_k a_{jk}}
$$
Construct the matrix $\mtx{A} = [ a_{jk} ]$.  As we will see, the variation of $Z$ is controlled by spectral properties of the matrix $\mtx{A}$.

For that purpose, we compute the Frobenius norm and spectral norm of $\mtx{A}$.
Each of its entries satisfies
$$
\abs{a_{jk}} \leq \frac{1}{W} \eta_{jk}
   = \frac{1}{W} (\mtx{H}^\adj\mtx{H})_{jk}
$$
owing to the fact that the entries of $\mtx{F}$ are uniformly bounded by $W^{-1/2}$.  The structure of $\mtx{H}$ implies that $\mtx{H}^\adj \mtx{H}$ is a symmetric, 0--1 matrix with exactly $W/R$ nonzero entries in each of its $W$ rows.  Therefore,
$$
\fnorm{\mtx{A}}
   \leq \frac{1}{W} \fnorm{ \mtx{H}^\adj \mtx{H} }
   = \frac{1}{W} \sqrt{W \cdot \frac{W}{R}}
   = \frac{1}{\sqrt{R}}.
$$

The spectral norm of $\mtx{A}$ is bounded by the spectral norm of its entrywise absolute value $\operatorname{abs}(\mtx{A})$.  Applying the fact~\eqref{eqn:H-norm} that $\norm{\mtx{H}} = \sqrt{W/R}$, we obtain
$$
\norm{\mtx{A}}
   \leq \norm{ \operatorname{abs}(\mtx{A}) }
   \leq \frac{1}{W} \norm{ \mtx{H}^\adj \mtx{H} }
   = \frac{1}{W} \normsq{\mtx{H}}
   = \frac{1}{R}.
$$
Let us emphasize that the norm bounds are uniform over all pairs $(\alpha, \omega)$ of frequencies.  Moreover, the matrix $\mtx{B} = [b_{jk}]$ satisfies identical bounds.

To continue, we estimate the mean of the variable $Z$.  This calculation is a simple consequence of H{\"o}lder's inequality:
\begin{align*}
\Expect Z &\leq (\Expect Z^2)^{1/2}
   = \left[ \Expect \abssq{ \sum\nolimits_{j} \eps_j'
       \left(\sum\nolimits_{k} \eps_k a_{jk} \right) } \right]^{1/2} \\
   &= \left[ \Expect \sum\nolimits_j \abssq{ \sum\nolimits_k \eps_k a_{jk} } \right]^{1/2}
   = \left[ \sum\nolimits_{j,k} \abssq{a_{jk}} \right]^{1/2} \\
   &= \fnorm{\mtx{A}}
   = \frac{1}{\sqrt{R}}.
\end{align*}

Chaos random variables, such as $Z$, concentrate sharply about their mean.  Deviations of $Z$ from its mean are controlled by two separate variances.  The probability of large deviation is determined by
$$
U = \sup\nolimits_{\enorm{\vct{u}}=1} \abs{
   \sum\nolimits_{j, k} u_j u_k a_{jk}},
$$
while the probability of a moderate deviation depends on
$$
V = \Expect \sup\nolimits_{\enorm{\vct{u}}=1} \abs{ \sum\nolimits_{j, k} u_j \eps_k a_{jk}}.
$$
To compute the first variance $U$, observe that
$$
U  = \norm{\mtx{A}}
   \leq \frac{1}{R}.
$$
The second variance is not much harder.  Using Jensen's inequality, we find
\begin{align*}
V &= \Expect \left[ \sum\nolimits_j \abssq{ \sum\nolimits_k \eps_k a_{jk} } \right]^{1/2} \\
   &\leq \left[\Expect \sum\nolimits_j \abssq{ \sum\nolimits_k \eps_k a_{jk} } \right]^{1/2}
   = \frac{1}{\sqrt{R}}.
\end{align*}

We are prepared to appeal to the following theorem which bounds the moments of a chaos variable~\cite[Cor.~2]{Ada05:Logarithmic-Sobolev}.

\begin{thm}[Moments of chaos] \label{thm:chaos}
Let $Z$ be a decoupled, symmetric, second-order chaos.  Then
$$
\Expect^p \abs{Z - \Expect Z} \leq \cnst{K} \left[ \sqrt{p}V + pU \right]
$$
for each $p \geq 1$.
\end{thm}

Substituting the values of the expectation and variances, we reach
$$
\Expect^p \left[ Z - \frac{1}{\sqrt{R}} \right]
   \leq \cnst{K}\left[\sqrt{\frac{p}{R}} + \frac{p}{R} \right].
$$
The content of this estimate is clearer, perhaps, if we split the bound at $p = R$.
$$
\Expect^p \left[ Z - \frac{1}{\sqrt{R}} \right]
   \leq \begin{cases}
   2\cnst{K} \sqrt{p/R}, \quad & p \leq R \\
   2\cnst{K} p/R, \quad & p > R.
   \end{cases}
$$
In words, the variable $Z$ exhibits subgaussian behavior in the moderate deviation regime, but its tail is subexponential.

We are now prepared to bound the moments of the maximum of the ensemble $\{Z_{\alpha\omega}\}$.  When $p = 2 \log W \leq R$, inequality~\eqref{eqn:moment-max} yields
\begin{align*}
\Expect^p \max_{\alpha, \omega} \left[Z_{\alpha\omega} - \frac{1}{\sqrt{R}} \right]
   &\leq W^{2/p} \max_{\alpha, \omega} \Expect^p
       \left[ Z_{\alpha\omega} - \frac{1}{\sqrt{R}} \right] \\
   &\leq \econst \cdot 2 \cnst{K} \sqrt{\frac{p}{R}},
\end{align*}
Recalling the definitions of $Y$ and $Z_{\alpha\omega}$, we reach
$$
\Expect^p Y
   = \Expect^p \max_{\alpha, \omega} Z_{\alpha\omega}
   \leq \cnst{C} \sqrt{\frac{\log W}{R}}.
$$
In words, we have obtained a moment bound for the decoupled version of the real chaos $M_{\real}$.

To complete the proof, remember that $\Expect^p M_{\real} \leq 4\Expect^p Y$.  Therefore,
$$
\Expect^p M_{\real} \leq 4\cnst{C}\sqrt{\frac{\log W}{R}}.
$$
An identical argument establishes that
$$
\Expect^p M_{\imag} \leq 4\cnst{C}\sqrt{\frac{\log W}{R}}.
$$
Since $M \leq M_{\real} + M_{\imag}$, it follows inexorably that
$$
\Expect^p M \leq 8\cnst{C} \sqrt{\frac{\log W}{R}}.
$$

Finally, we invoke Markov's inequality to obtain the tail bound
$$
\Prob{ M \geq \econst^{0.5} \cdot 8\cnst{C} \sqrt{\frac{\log W}{R}} }
   \leq \econst^{-0.5p}
   = W^{-1}.
$$
This endeth the lesson.
\end{proof}

\subsection{Column Norms and Coherence}

Lemma~\ref{lem:X-max-bd} has two important and immediate consequences.  First, it implies that the columns of $\Fee$ essentially have unit norm.

\begin{thm}[Column Norms] \label{thm:col-norms}
Suppose the sampling rate
$$
R \geq \cnst{C} \delta^{-2} \log W.
$$
An $R \times W$ random demodulator $\Fee$ satisfies
$$
\Prob{ \max\nolimits_\omega \abs{ \enormsq{\atom_\omega} - 1 } \geq \delta }
   \leq W^{-1}.
$$
\end{thm}

\vspace{.5pc}

Lemma~\ref{lem:X-max-bd} also shows that the \term{coherence} of the random demodulator is small.  The coherence, which is denoted by $\mu$, bounds the maximum inner product between distinct columns of $\Fee$, and it has emerged as a fundamental tool for establishing the success of compressive sampling recovery algorithms. Rigorously,
$$
\mu = \max_{\alpha\neq\omega} \absip{ \atom_\alpha }{ \atom_\omega }.
$$
We have the following probabilistic bound.

\begin{thm}[Coherence] \label{thm:coherence}
Suppose the sampling rate
$$
R \geq 2 \log W.
$$
An $R \times W$ random demodulator $\Fee$ satisfies
$$
\Prob{ \mu \geq \cnst{C} \sqrt{\frac{\log W}{R}} } \leq W^{-1}.
$$
\end{thm}

\vspace{.5pc}

For a general $R \times W$ matrix with unit-norm columns, we can verify~\cite[Thm.~2.3]{SH03:Grassmannian-Frames} that its coherence
$$
\mu \geq \frac{1}{\sqrt{R}} \left[ 1 - \frac{R}{W} \right] \approx \frac{1}{\sqrt{R}}.
$$
Since the columns of the random demodulator are essentially normalized, we conclude that its coherence is nearly as small as possible.

\section{Recovery for the Random Phase Model} \label{app:random-phase}

In this appendix, we establish Theorem~\ref{thm:rpm}, which shows that $\ell_1$ minimization is likely to recover a random signal drawn from Model (A).  Appendix~\ref{app:rip} develops results for general signals.

The performance of $\ell_1$ minimization for random signals depends on several subtle properties of the demodulator matrix.  We must discuss these ideas in some detail.  In the next two sections, we use the bounds from Appendix~\ref{app:rand-demod-mtx} to check that the required conditions are in force.  Afterward, we proceed with the demonstration of Theorem~\ref{thm:rpm}.

\subsection{Cumulative Coherence}

The coherence measures only the inner product between a single pair of columns.  To develop accurate results on the performance of $\ell_1$ minimization, we need to understand the total correlation between a fixed column and a collection of distinct columns.

Let $\Fee$ be an $R \times W$ matrix, and let $\Omega$ be a subset of $\oneton{W}$.  The \term{local 2-cumulative coherence} of the set $\Omega$ with respect to the matrix $\Fee$ is defined as
$$
\mu_2(\Omega) = \max_{\alpha \notin \Omega}
   \left[ \sum\nolimits_{\omega \in \Omega}
   \abssqip{ \atom_\alpha }{ \atom_\omega} \right]^{1/2}.
$$
The coherence $\mu$ provides an immediate bound on the cumulative coherence:
\begin{equation*} \label{eqn:mu2-simple-bd}
\mu_2(\Omega) \leq \mu \sqrt{ \abs{\Omega} }.
\end{equation*}
Unfortunately, this bound is completely inadequate.

To develop a better estimate, we instate some additional notation.  Consider the matrix norm $\pnorm{1\to2}{\cdot}$, which returns the maximum $\ell_2$ norm of a column.  Define the \term{hollow Gram matrix}
\begin{equation*} \label{eqn:hollow-gram}
\mtx{G} = \Fee^\adj \Fee - \diag( \Fee^\adj \Fee ),
\end{equation*}
which tabulates the inner products between \emph{distinct} columns of $\Fee$.  Let $\mtx{R}_{\Omega}$ be the $W \times W$ orthogonal projector onto the coordinates listed in $\Omega$.  Elaborating on these definitions, we see that
\begin{equation} \label{eqn:mu2-bd}
\mu_2(\Omega)
   = \pnorm{1\to2}{ \mtx{R}_\Omega \mtx{G} (\Id - \mtx{R}_{\Omega}) }
   \leq \pnorm{1\to2}{ \mtx{R}_{\Omega} \mtx{G} }.
\end{equation}
In particular, the cumulative coherence $\mu_2(\Omega)$ is dominated by the maximum column norm of $\mtx{G}$.

When $\Omega$ is chosen at random, it turns out that we can use this upper bound to improve our estimate of the cumulative coherence $\mu_2(\Omega)$.  To incorporate this effect, we need the following result, which is a consequence of Theorem~3.2 of \cite{Tro08:Norms-Random} combined with a standard decoupling argument (e.g., see \cite[Lem.~14]{Tro08:Linear-Independence}).

\begin{prop} \label{prop:12-norm}
Fix a $W \times W$ matrix $\mtx{G}$.  Let $\mtx{R}$ be an orthogonal projector onto $K$ coordinates, chosen randomly from $\oneton{W}$.  For $p = 2\log W$,
$$
\Expect^p \pnorm{1\to2}{\mtx{RG}} \leq 8 \sqrt{\log W} \pnorm{\max}{\mtx{G}}
   + 2 \sqrt{\frac{K}{W}} \pnorm{1\to2}{\mtx{G}}.
$$
\end{prop}

\vspace{.5pc}

With these facts at hand, we can establish the following bound.

\begin{thm}[Cumulative Coherence] \label{thm:cumulative-coherence}
Suppose the sampling rate
$$
R \geq \cnst{C} \left[ K \log W + \log^3 W \right].
$$
Draw an $R \times W$ random demodulator $\Fee$.  Let $\Omega$ be a random set of $K$ coordinates in $\oneton{W}$.  Then
$$
\Prob{ \mu_2(\Omega) \geq \frac{1}{\sqrt{16\log W}} } \leq 3W^{-1}.
$$
\end{thm}

\vspace{.5pc}

\begin{proof}
Under our hypothesis on the sampling rate, Theorem~\ref{thm:coherence} demonstrates that, except with probability $W^{-1}$,
$$
\pnorm{\max}{\mtx{G}} \leq \frac{\cnst{c}}{\log W},
$$
where we can make $\cnst{c}$ as small as we like by increasing the constant in the sampling rate.  Similarly, Theorem~\ref{thm:col-norms} ensures that
$$
\max_\omega \enorm{\atom_\omega} \leq 2,
$$
except with probability $W^{-1}$.  We condition on the event $F$ that these two bounds hold.  We have $\Prob{F^c} \leq 2W^{-1}$.

On account of the latter inequality and the fact \eqref{eqn:spec-norm} that $\norm{\Fee} = \sqrt{W/R}$, we obtain
\begin{multline*}
\pnorm{1\to2}{\mtx{G}} \leq \pnorm{1\to2}{\Fee^\adj \Fee}
   = \max_\omega \enorm{ \Fee^\adj \Fee \onevct_\omega } \\
   = \max_\omega \enorm{ \Fee^\adj \atom_\omega }
   \leq \norm{\Fee} \cdot \max_{\omega} \enorm{\atom_\omega}
   \leq 2\sqrt{W/R}.
\end{multline*}
We have written $\onevct_\omega$ for the $\omega$th standard basis vector in $\Cspace{W}$.

Let $\mtx{R}_{\Omega}$ be the (random) projector onto the $K$ coordinates listed in $\Omega$.  For $p = 2\log W$, relation~\eqref{eqn:mu2-bd} and Proposition~\ref{prop:12-norm} yield
\begin{align*}
\Expect^p [ \mu_2( \Omega ) \, |\, F ]
   &\leq \Expect^p [ \pnorm{1\to2}{\mtx{R}_{\Omega} \mtx{G}} \, | \, F] \\
   &\leq 8\sqrt{\log W} \cdot \frac{\cnst{c}}{\log W}
       + 2\sqrt{\frac{K}{W}} \cdot 2\sqrt{\frac{W}{R}} \\
   &\leq \frac{\cnst{c}'}{\sqrt{\log W}}
\end{align*}
under our assumption on the sampling rate.  Finally, we apply Markov's inequality to obtain
$$
\Prob{ \mu_2( \Omega ) > \frac{\cnst{c}'\econst^{0.5}}{\sqrt{\log W}} \, \big| \, F }
   \leq \econst^{-0.5 p} = W^{-1}.
$$
By selecting the constant in the sampling rate sufficiently large, we can ensure $\cnst{c}'$ is sufficiently small that
$$
\Prob{ \mu_2(\Omega) > \frac{1}{\sqrt{16\log W}} \, | \, F } \leq W^{-1}.
$$
We reach the conclusion
$$
\Prob{ \mu_2(\Omega) > \frac{1}{\sqrt{16\log W}} }
   \leq W^{-1} + \Prob{F^c} \leq 3W^{-1}
$$
when we remove the conditioning.
\end{proof}

\subsection{Conditioning of a Random Submatrix}

We also require information about the conditioning of a random set of columns drawn from the random demodulator matrix $\Fee$.  To obtain this intelligence, we refer to the following result, which is a consequence of~\cite[Thm.~1]{Tro08:Norms-Random} and a decoupling argument.

\begin{prop} \label{prop:rdm-spec-norm}
Let $\mtx{A}$ be a $W \times W$ Hermitian matrix, split into its diagonal and off-diagonal parts: $\mtx{A} = \mtx{E} + \mtx{G}$.  Draw an orthogonal projector $\mtx{R}$ onto $K$ coordinates, chosen randomly from $\oneton{W}$. For $p = 2\log W$,
\begin{multline*}
\Expect^p \norm{\mtx{RAR}}
   \leq  \cnst{C} \bigg[ \log W \pnorm{\max}{\mtx{A}} \\
       + \sqrt{\frac{K\log W}{W}} \pnorm{1\to2}{\mtx{A}}
       + \frac{K}{W} \norm{ \mtx{A} } \bigg] + \norm{\mtx{E}}.
\end{multline*}
\end{prop}

\vspace{.5pc}

Suppose that $\Omega$ is a subset of $\oneton{W}$.  Recall that $\Fee_{\Omega}$ is the column submatrix of $\Fee$ indexed by $\Omega$.  We have the following result.

\begin{thm}[Conditioning of a Random Submatrix] \label{thm:cond-rdm}
Suppose the sampling rate
$$
R \geq \cnst{C} \left[ K \log W + \log^3 W \right].
$$
Draw an $R \times W$ random demodulator, and let $\Omega$ be a random subset of $\oneton{W}$ with cardinality $K$.  Then
$$
\Prob{ \norm{ \Fee_\Omega^\adj \Fee_\Omega - \Id } \geq 0.5 } \leq 3W^{-1}.
$$
\end{thm}

\vspace{.5pc}

\begin{proof}
Define the quantity of interest
$$
Q = \norm{ \Fee_\Omega^\adj \Fee_\Omega - \Id }.
$$
Let $\mtx{R}_{\Omega}$ be the random orthogonal projector onto the coordinates listed in $\Omega$, and observe that
$$
Q = \norm{ \mtx{R}_{\Omega} ( \Fee^\adj \Fee - \Id ) \mtx{R}_\Omega }.
$$
Define the matrix $\mtx{A} = \Fee^\adj \Fee - \Id$.  Let us perform some background investigations so that we are fully prepared to apply Proposition~\ref{prop:rdm-spec-norm}.

We split the matrix
$$
\mtx{A} = \mtx{E} + \mtx{G},
$$
where $\mtx{E} = \diag(\Fee^\adj \Fee) - \Id$ and $\mtx{G} = \Fee^\adj \Fee - \diag(\Fee^\adj \Fee)$ is the hollow Gram matrix.  Under our hypothesis on the sampling rate, Theorem~\ref{thm:col-norms} provides that
$$
\norm{\mtx{E}} = \max_{\omega} \abs{\enormsq{\atom_\omega} - 1}
   \leq 0.15,
$$
except with probability $W^{-1}$.  It also follows that
$$
\norm{\diag(\Fee^\adj \Fee)} = \max_\omega \enormsq{\atom_\omega} \leq 1.15.
$$
We can bound $\pnorm{1\to2}{\mtx{G}}$ by repeating our earlier calculation and introducing the identity \eqref{eqn:spec-norm} that $\norm{\Fee} = \sqrt{W/R}$.  Thus,
$$
\pnorm{1\to2}{\mtx{G}} \leq \norm{\Fee} \max_\omega \enorm{\atom_\omega}
   \leq \sqrt{\frac{1.15 W}{R}}.
$$
Since $W/R \geq 1$, it also holds that
$$
\norm{\mtx{G}}
   = \norm{\Fee^\adj \Fee - \diag(\Fee^\adj \Fee)}
   \leq \normsq{\Fee} + 1.15
   \leq \frac{2.15W}{R}.
$$
Meanwhile, Theorem~\ref{thm:coherence} ensures that
$$
\pnorm{\max}{ \mtx{G} } \leq \frac{\cnst{c}}{\log W},
$$
except with probability $W^{-1}$.  As before, we can make $\cnst{c}$ as small as we like by increasing the constant in the sampling rate.  We condition on the event $F$ that these estimates are in force.  So $\Prob{F^c} \leq 2W^{-1}$.

Now, invoke Proposition~\ref{prop:rdm-spec-norm} to obtain
\begin{multline*}
\Expect^p [ Q \, | \, F ]
   \leq \cnst{C} \bigg[ \log W \cdot \frac{\cnst{c}}{\log W} \\
       + \sqrt{\frac{K\log W}{W}} \cdot \sqrt{\frac{1.15W}{R}}
       + \frac{K}{W} \cdot \frac{1.15W}{R} \bigg]
       + 0.15
\end{multline*}
for $p = 2\log W$.  Simplifying this bound, we reach
$$
\Expect^p [ Q \, | \, F ]
   \leq \cnst{C} \left[ \cnst{c}
       + \sqrt{\frac{\cnst{C}' K \log W}{R}} \right] + 0.15.
$$
By choosing the constant in the sampling rate $R$ sufficiently large, we can guarantee that
$$
\Expect^p [ Q \, | \, F] \leq 0.3.
$$
Markov's inequality now provides
$$
\Prob{ Q \geq 0.3\econst^{0.5} \, | \, F }
   \leq \econst^{-0.5p}
   = W^{-1}.
$$
Note that $0.3\econst^{0.5} < 0.5$ to conclude that
$$
\Prob{ Q \geq 0.5 } \leq W^{-1} + \Prob{ F^c } \leq 3W^{-1}.
$$
This is the advertised conclusion.
\end{proof}

\subsection{Recovery of Random Signals}

The literature contains powerful results on the performance of $\ell_1$ minimization for recovery of random signals that are sparse with respect to an incoherent dictionary.  We have finally acquired the keys we need to start this machinery.  Our major result for random signals, Theorem~\ref{thm:rpm}, is a consequence of the following result, which is an adaptation of~\cite[Thm.~14]{Tro08:Conditioning-Random}.

\begin{prop} \label{prop:l1-rdm-recover}
Let $\Fee$ be an $R \times W$ matrix, and let $\Omega$ be a subset of $\oneton{W}$ for which
\begin{itemize}
\item   $\mu_2(\Omega) \leq (16\log W)^{-1/2}$, and
\item   $\norm{ (\Fee_{\Omega}^\adj \Fee_\Omega)^{-1} } \leq 2$.
\end{itemize}
Suppose that $\vct{s}$ is a vector that satisfies
\begin{itemize}
\item   $\supp{\vct{s}} \subset \Omega$, and
\item   $\sgn(\vct{s}_\omega)$ are i.i.d.~uniform on the complex unit circle.
\end{itemize}
Then $\vct{s}$ is the unique solution to the optimization problem
\begin{equation} \label{eqn:p1-app}
\min \ \pnorm{1}{\vct{v}}
\quad\subjto\quad
\Fee \vct{v} = \Fee \vct{s}
\end{equation}
except with probability $2W^{-1}$.
\end{prop}

Our main result, Theorem~\ref{thm:rpm}, is a corollary of this result.

\begin{cor}
Suppose that the sampling rate satisfies
\begin{equation*}
\label{eqn:Rrpm}
R \geq \cnst{C} \left[ K \log W + \log^3 W \right].
\end{equation*}
Draw an $R\times W$ random demodulator $\mtx{\Fee}$.  Let $\vct{s}$ be a random amplitude vector drawn according to Model (A).  Then the solution $\widehat{\vct{s}}$ to the convex program \eqref{eqn:p1-app} equals $\vct{s}$ except with probability $8W^{-1}$.
\end{cor}

\begin{proof}
Since the signal $\vct{s}$ is drawn from Model (A), its support $\Omega$ is a uniformly random set of $K$ components from $\oneton{W}$.  Theorem~\ref{thm:cumulative-coherence} ensures that
$$
\mu_2(\Omega) \leq \frac{1}{\sqrt{16\log W}},
$$
except with probability $3W^{-1}$.  Likewise, Theorem~\ref{thm:cond-rdm} guarantees
$$
\norm{ \Fee_\Omega^\adj \Fee_\Omega - \Id } \leq 0.5,
$$
except with probability $3W^{-1}$.  This bound implies that all the eigenvalues of $\Fee_\Omega^\adj \Fee_\Omega$ lie in the range $[0.5, 1.5]$.  In particular,
$$
\norm{ (\Fee_\Omega^\adj \Fee_\Omega)^{-1} } \leq 2.
$$
Meanwhile, Model (A) provides that the phases of the nonzero entries of $\vct{s}$ are uniformly distributed on the complex unit circle.  We invoke Proposition~\ref{prop:l1-rdm-recover} to see that the optimization problem \eqref{eqn:p1-app} recovers the amplitude vector $\vct{s}$ from the observations $\Fee\vct{s}$, except with probability $2W^{-1}$.  The total failure probability is $8W^{-1}$.
\end{proof}

\section{Stable Recovery}
\label{app:rip}

In this appendix, we establish much stronger recovery results under a slightly stronger assumption on the sampling rate.  This work is based on the \term{restricted isometry property} (RIP), which enjoys the privilege of a detailed theory.

\subsection{Background}

The RIP is a formalization of the statement that the sampling matrix preserves the norm of all sparse vectors up to a small constant factor.  Geometrically, this property ensures that the sampling matrix embeds the set of sparse vectors in a high-dimensional space into the lower-dimensional space of samples.  Consequently, it becomes possible to recover sparse vectors from a small number of samples.  Moreover, as we will see, this property also allows us to approximate compressible vectors, given relatively few samples.

Let us shift to rigorous discussion.  We say that an $R \times W$
matrix $\Fee$ has the RIP of order $N$ with restricted isometry
constant $\delta_N \in (0,1)$ when the following inequalities are in
force:
\begin{equation} \label{eqn:rip-1}
\abs{ \frac{\enormsq{ \Fee\vct{x} } - \enormsq{\vct{x}}}{\enormsq{\vct{x}}} }
   \leq \delta_N
\quad\text{whenever}\quad
\pnorm{0}{\vct{x}} \leq N.
\end{equation}
Recall that the function $\pnorm{0}{\cdot}$ counts the number of nonzero entries in a vector.  In words, the definition states that the sampling operator produces a small relative change in the squared $\ell_2$ norm of an $N$-sparse vector.  The RIP was introduced by Cand{\`e}s and Tao in an important paper~\cite{CT06:Near-Optimal}.

For our purposes, it is more convenient to rephrase the RIP.  Observe that the inequalities \eqref{eqn:rip-1} hold if and only if
$$
\abs{ \frac{ \vct{x}^\adj ( \Fee^\adj \Fee - \Id ) \vct{x} }{ \vct{x}^\adj \vct{x} } }
   \leq \delta_N
\quad\text{whenever}\quad
\pnorm{0}{\vct{x}} \leq N.
$$
The extreme value of this Rayleigh quotient is clearly the largest magnitude eigenvalue of any $N \times N$ principal submatrix of $\Fee^\adj \Fee - \Id$.

Let us construct a norm that packages up this quantity.  Define
$$
\triplenorm{\mtx{A}} = \sup_{\abs{\Omega}\leq N} \norm{ \mtx{A}\restrict{\Omega \times \Omega} }.
$$
In words, the triple-bar norm returns the least upper bound on the spectral norm of any $N \times N$ principal submatrix of $\mtx{A}$.  Therefore, the matrix $\Fee$ has the RIP of order $N$ with constant $\delta_N$ if and only if
$$
\triplenorm{ \Fee^\adj \Fee - \Id } \leq \delta_N.
$$
Referring back to \eqref{eqn:gram-expect}, we may write this relation in the more suggestive form
$$
\triplenorm{ \Fee^\adj \Fee - \Expect \Fee^\adj \Fee } \leq \delta_N.
$$
In the next section, we strive to achieve this estimate.

\subsection{RIP for Random Demodulator}

Our major result is that the random demodulator matrix has the restricted isometry property when the sampling rate is chosen appropriately.

\begin{thm}[RIP for Random Demodulator] \label{thm:demod-rip}
Fix $\delta > 0$.  Suppose that the sampling rate
$$
R \geq \cnst{C} \delta^{-2} \cdot N \log^6(W).
$$
Then an $R \times W$ random demodulator matrix $\Fee$ has the restricted isometry property of order $N$ with constant $\delta_{N} \leq \delta$, except with probability $\bigO(W^{-1})$.
\end{thm}

There are three basic reasons that the random demodulator matrix satisfies the RIP. First, its Gram matrix averages to the identity.  Second, its rows are independent.  Third, its entries are uniformly bounded.  The proof shows how these three factors interact to produce the conclusion.

In the sequel, it is convenient to write $\vct{z} \otimes \vct{z}$ for the rank-one matrix $\vct{z}\vct{z}^\adj$.  We also number constants for clarity.

Most of the difficulty of argument is hidden in a lemma due to Rudelson and Vershynin~\cite[Lem.~3.6]{RV06:Sparse-Reconstruction}.

\begin{lemma}[Rudelson--Vershynin]
Suppose that $\{ \vct{z}_r \}$ is a sequence of $R$ vectors in $\Cspace{W}$ where $R \leq W$, and assume that each vector satisfies the bound $\infnorm{ \vct{z}_r } \leq B$.  Let $\{\xi_r\}$ be an independent Rademacher series.  Then
$$
\Expect
\triplenorm{ \sum\nolimits_{r} \xi_r \vct{z}_r \otimes \vct{z}_r } \\
\leq \beta \triplenorm{ \sum\nolimits_{r} \vct{z}_r \otimes \vct{z}_r }^{1/2},
$$
where
$$
\beta \leq \cnst{C}_1 B \sqrt{N} \log^2(W).
$$
\end{lemma}

\vspace{.5pc}

The next lemma employs this bound to show that, in expectation, the random demodulator has the RIP.

\begin{lemma} \label{lem:demod-expect-rip}
Fix $\delta \in (0,1)$.  Suppose that the sampling rate
$$
R \geq \cnst{C}_2 \delta^{-2} \cdot N \log^5 W.
$$
Let $\Fee$ be an $R \times W$ random demodulator.  Then
$$
\Expect \triplenorm{\Fee^\adj \Fee - \Id } \leq \delta.
$$
\end{lemma}

\vspace{.5pc}

\begin{proof}
Let $\vct{z}_r^\adj$ denote the $r$th row of $\Fee$.  Observe that the rows are mutually \emph{independent} random vectors, although their distributions differ.  Recall that the entries of each row are uniformly bounded, owing to Lemma~\ref{lem:entry-bd}.
We may now express the Gram matrix of $\Fee$ as
$$
\Fee^\adj \Fee
   = \sum\nolimits_{r=1}^R \vct{z}_r \otimes \vct{z}_r.
$$
We are interested in bounding the quantity
\begin{align*}
E = \Expect \triplenorm{ \Fee^\adj\Fee - \Id }
   &= \Expect \triplenorm{ \sum\nolimits_r \vct{z}_r \otimes \vct{z}_r - \Id } \\
   &= \Expect \triplenorm{ \sum\nolimits_r \left( \vct{z}_r \otimes \vct{z}_r
   - \Expect \vct{z}_r' \otimes \vct{z}_r' \right) },
\end{align*}
where $\{\vct{z}_r'\}$ is an independent copy of $\{\vct{z}_r\}$.

The symmetrization result, Proposition~\ref{prop:symmetrization}, yields
$$
E \leq 2 \Expect
   \triplenorm{ \sum\nolimits_r \xi_r \vct{z}_r \otimes \vct{z}_r },
$$
where $\{ \xi_r \}$ is a Rademacher series, independent of everything else.  Conditional on the choice of $\{\vct{z}_r\}$, the Rudelson--Vershynin lemma results in
$$
E \leq 2 \Expect \left( \beta \cdot \triplenorm{
   \sum\nolimits_r \vct{z}_r \otimes \vct{z}_r }^{1/2} \right),
$$
where
$$
\beta \leq \cnst{C}_1 B \sqrt{N} \log^2 W
$$
and $B = \max_{k,\omega} \abs{\varphi_{k\omega}}$ is a random variable.

According to Lemma~\ref{lem:entry-bd},
$$
( \Expect \beta^2 )^{1/2} \leq \sqrt{6}\cnst{C}_1 \sqrt{\frac{N\log^5(W)}{R}}.
$$
We may now apply the Cauchy--Schwarz inequality to our bound on $E$ to reach
$$
E \leq 2\sqrt{6}\cnst{C}_1\sqrt{\frac{ N \log^5 W}{R}} \left( \Expect
   \triplenorm{ \sum\nolimits_r \vct{z}_r \otimes \vct{z}_r } \right)^{1/2}.
$$
Add and subtract an identity matrix inside the norm, and invoke the triangle inequality.  Note that $\triplenorm{\Id} = 1$, and identify a copy of $E$ on the right-hand side:
\begin{align*}
E &\leq \sqrt{\frac{24\cnst{C}_1^2 N \log^5 W}{R}} \left( \Expect
   \triplenorm{ \sum\nolimits_r \vct{z}_r \otimes \vct{z}_r - \Id } +
   \triplenorm{\Id} \right)^{1/2} \\
   &= \sqrt{\frac{24\cnst{C}_1^2 N \log^5 W}{R}} (E + 1)^{1/2}.
\end{align*}
Solutions to the relation $E \leq \gamma (E + 1)^{1/2}$ obey $E \leq 2\gamma$ whenever $\gamma \leq 1$.

We conclude that
$$
E \leq \sqrt{\frac{\cnst{C}_2 N \log^5 W}{R}}
$$
whenever the fraction is smaller than one.  To guarantee that $E \leq \delta$, then, it suffices that
$$
R \geq \cnst{C}_2 \delta^{-2} \cdot N \log^5 W.
$$
This point completes the argument.
\end{proof}

To finish proving the theorem, it remains to develop a large deviation bound.  Our method is the same as that of Rudelson and Vershynin: We invoke a concentration inequality for sums of independent random variables in a Banach space.  See~\cite[Thm.~3.8]{RV06:Sparse-Reconstruction}, which follows from~\cite[Thm.~6.17]{LT91:Probability-Banach}.

\begin{prop} \label{prop:banach-tail}
Let $Y_1, \dots, Y_R$ be independent, symmetric random variables in a Banach space $X$, and assume each random variable satisfies the bound $\pnorm{X}{ Y_r } \leq B$ almost surely.  Let $Y = \pnorm{X}{\sum_r Y_r}$. Then
$$
\Prob{ Y > \cnst{C}_3 \left[ u \Expect Y + tB \right] }
   \leq \econst^{-u^2} + \econst^{-t}
$$
for all $u, t \geq 1$.
\end{prop}

In words, the norm of a sum of independent, symmetric, bounded random variables has a tail bound consisting of two parts.  The norm exhibits subgaussian decay with ``standard deviation'' comparable to its mean, and it exhibits subexponential decay controlled by the size of the summands.

\begin{proof}[Theorem~\ref{thm:demod-rip}]
We seek a tail bound for the random variable
$$
Z = \triplenorm{ \sum\nolimits_r \vct{z}_r \otimes \vct{z}_r - \Id }
   = \triplenorm{ \sum\nolimits_r (\vct{z}_r \otimes \vct{z}_r
        - \Expect \vct{z}_r' \otimes \vct{z}_r' ) }.
$$
As before, $\vct{z}_r^\adj$ is the $r$th row of $\Fee$ and $\vct{z}_r'$ is an independent copy of $\vct{z}_r$.

To begin, we express the constant in the stated sampling rate as $\cnst{C} = \cnst{C}_2 \cdot \cnst{c}^{-2}$, where $\cnst{c}$ will be adjusted at the end of the argument.  Therefore, the sampling rate may be written as
$$
R \geq \cnst{C}_2 \left( \frac{\cnst{c}\delta}{\sqrt{\log W}}\right)^{-2} N \log^5 W.
$$
Lemma~\ref{lem:demod-expect-rip} implies that
$$
\Expect Z \leq \frac{\cnst{c} \delta}{\sqrt{\log W}} \leq \cnst{c}\delta.
$$

As described in the background section, we can produce a tail bound for $Z$ by studying the symmetrized random variable
\begin{align*}
Y &= \triplenorm{ \sum\nolimits_r (\vct{z}_r \otimes \vct{z}_r
   - \vct{z}_r' \otimes \vct{z}_r') } \\
&\sim \triplenorm{ \sum\nolimits_r \xi_r (\vct{z}_r \otimes \vct{z}_r
   - \vct{z}_r' \otimes \vct{z}_r') },
\end{align*}
where $\{\xi_r\}$ is a Rademacher sequence independent of everything.
For future reference, use the first representation of $Y$ to see that
$$
\Expect Y \leq 2 \Expect Z \leq \frac{2\cnst{c} \delta}{\sqrt{\log W}},
$$
where the inequality follows from the triangle inequality and identical distribution.

Observe that $Y$ is the norm of a sum of independent, symmetric random variables in a Banach space.  The tail bound of Proposition~\ref{prop:banach-tail} also requires the summands to be bounded in norm.  To that end, we must invoke a truncation argument.  Define the (independent) events
$$
F_r = \left\{ \max \left\{\infnorm{\vct{z}_r}^2, \infnorm{\vct{z}_r'}^2 \right\}
   \leq \frac{10\log W}{R} \right\}
$$
and
$$
F = \bigcap\nolimits_r F_r.
$$
In other terms, $F$ is the event that two independent random demodulator matrices both have bounded entries. Therefore,
$$
\Prob{F^c} \leq 2W^{-1}
$$
after two applications of Lemma~\ref{lem:entry-bd}.

Under the event $F$, we have a hard bound on the triple-norm of each summand in $Y$.
\begin{align*}
B &= \max\nolimits_r \triplenorm{ \vct{z}_r \otimes \vct{z}_r - \vct{z}_r' \otimes \vct{z}_r'} \\
   &\leq \max\nolimits_r \left( \triplenorm{ \vct{z}_r \otimes \vct{z}_r }
       + \triplenorm{ \vct{z}_r' \otimes \vct{z}_r' } \right).
\end{align*}
For each $r$, we may compute
\begin{multline*}
\triplenorm{ \vct{z}_r \otimes \vct{z}_r }
   = \sup_{\abs{\Omega} \leq N } \norm{ (\vct{z}_r \otimes \vct{z}_r)\restrict{\Omega \times \Omega} } \\
   = \sup_{\abs{\Omega} \leq N} \enormsq{ {\vct{z}_r}\restrict\Omega }
   \leq N \infnorm{\vct{z}_r}^2
   \leq \frac{10N \log W}{R},
\end{multline*}
where the last inequality follows from $F$.
An identical estimate holds for the terms involving $\vct{z}_r'$.  Therefore,
$$
B \leq \frac{20 N \log W}{R}.
$$
Our hypothesized lower bound for the sampling rate $R$ gives
$$
B \leq 20 N \log W \cdot \frac{(\cnst{c}\delta /\sqrt{\log W})^2}{\cnst{C}_2 N\log^5 W}
   \leq \frac{\cnst{c} \delta}{\log W},
$$
where the second inequality certainly holds if $\cnst{c}$ is a sufficiently small constant.

Now, define the truncated variable
$$
Y_{\rm trunc} = \triplenorm{ \sum\nolimits_r \xi_r (\vct{z}_r \otimes \vct{z}_r
   - \vct{z}_r' \otimes \vct{z}_r') \mathbb{I}_{F_r} }.
$$
By construction, the triple-norm of each summand is bounded by $B$.
Since $Y$ and $Y_{\rm trunc}$ coincide on the event $F$, we have
\begin{align*}
\Prob{ Y > v }
	&\leq \Prob{ Y > v \ | \ F } \cdot \Prob{F} + \Prob{F^c} \\
	&= \Prob{ Y_{\rm trunc} > v \ |\ F } \cdot \Prob{F} + \Prob{F^c} \\
	&= \Prob{ Y_{\rm trunc} > v } + \Prob{F^c}
\end{align*}
for each $v > 0$.
According to the contraction principle~\cite[Thm.~4.4]{LT91:Probability-Banach},
the bound
$$
\Expect_{\vct{\xi}} Y_{\rm trunc} \leq \Expect_{\vct{\xi}} Y
$$
holds pointwise for each choice of $\{\vct{z}_r\}$ and $\{\vct{z}_r'\}$.
Therefore,
$$
\Expect Y_{\rm trunc} = \Expect_{\vct{z}_r, \vct{z}_r'} \Expect_{\vct{\xi}} Y_{\rm trunc}
	\leq \Expect_{\vct{z}_r, \vct{z}_r'} \Expect_{\vct{\xi}} Y
	= \Expect Y.
$$
Recalling our estimate for $\Expect Y$, we see that
$$
\Expect Y_{\rm trunc}
	\leq \frac{2\cnst{c} \delta}{\sqrt{\log W}}.
$$

We now apply Proposition~\ref{prop:banach-tail} to the symmetric variable $Y_{\rm trunc}$. The bound reads
$$
\Prob{ Y_{\rm trunc} > \cnst{C}_3 \left( u \Expect Y_{\rm trunc} + tB \right) }
   \leq \econst^{-u^2} + \econst^{-t}.
$$
Select $u = \sqrt{\log W}$ and $t = \log W$, and recall the bounds on $\Expect Y_{\rm trunc}$ and on $B$ to obtain
$$
\Prob{ Y_{\rm trunc} > \cnst{C}_3 (2\cnst{c} \delta + \cnst{c}\delta) } \leq 2W^{-1}.
$$
Using the relationship between the tails of $Y$ and $Y_{\rm trunc}$, we reach
\begin{align*}
\Prob{ Y > 3\cnst{C}_3 \cnst{c} \delta }
   &\leq \Prob{ Y_{\rm trunc} > 3\cnst{C}_3 \cnst{c} \delta } + \Prob{F^c} \\
   &\leq 4W^{-1}.
\end{align*}

Finally, the tail bound for $Y$ yields a tail bound for $Z$ via relation~\eqref{eqn:symm-tail}:
$$
\Prob{ Z > 2 \Expect{Z} + u } \leq 2 \Prob{ Y > u }.
$$
As noted, $\Expect Z \leq \cnst{c}\delta \leq \cnst{C}_3 \cnst{c} \delta$.  Therefore,
\begin{align*}
\Prob{Z > 5\cnst{C}_3 \cnst{c} \delta}
   &\leq \Prob{Z > 2\Expect Z + 3\cnst{C}_3 \cnst{c}\delta} \\
   &\leq 2\Prob{Y > 3\cnst{C}_3 \cnst{c}\delta } \\
   &\leq 8W^{-1}.
\end{align*}
To complete the proof, we select $\cnst{c} \leq (5\cnst{C}_3)^{-1}$.
\end{proof}

\subsection{Signal Recovery under the RIP}

When the sampling matrix has the restricted isometry property, the samples contain enough information to approximate general signals extremely well.  Cand{\`e}s, Romberg, and Tao have shown that signal recovery can be performed by solving a convex optimization problem.  This approach produces exceptional results in theory and in practice~\cite{CRT06:Stable-Signal}.

\begin{prop} \label{prop:l1-stable}
Suppose that $\Fee$ is an $R \times W$ matrix that verifies the RIP of order $2K$ with restricted isometry constant $\delta_{2K} \leq \cnst{c}$.  Then the following statement holds.  Consider an arbitrary vector $\vct{s}$ in $\Cspace{W}$, and suppose we collect noisy samples
$$
\vct{y} = \Fee\vct{s} + \vct{\nu}
$$
where $\enorm{\vct{\nu}} \leq \eta$.  Every solution $\widehat{\vct{s}}$ to the optimization problem
\begin{equation} \label{eqn:l1-err-app}
\min \pnorm{1}{\vct{v}}
\quad\subjto\quad
\enorm{\Fee\vct{v} - \vct{y}} \leq \eta,
\end{equation}
approximates the target signal:
$$
\enorm{ \widehat{\vct{s}} - \vct{s} }
   \leq \cnst{C} \max\left\{ \eta,
       \frac{1}{\sqrt{K}} \pnorm{1}{\vct{s} - \vct{s}_{K}} \right\},
$$
where $\vct{s}_K$ is a best $K$-sparse approximation to $\vct{s}$ with respect to the $\ell_1$ norm.
\end{prop}

An equivalent guarantee holds when the approximation $\widehat{\vct{s}}$ is computed using the {\sf CoSaMP} algorithm~\cite[Thm.~A]{NT08:CoSaMP-Iterative}.

Combining Proposition~\ref{prop:l1-stable} with Theorem~\ref{thm:demod-rip}, we obtain our major result, Theorem~\ref{thm:stable-recovery}.

\begin{cor}
Suppose that the sampling rate
$$
R \geq \cnst{C} K \log^6 W.
$$
An $R \times W$ random demodulator matrix verifies the RIP of order $2K$ with constant $\delta_{2K} \leq \cnst{c}$, except with probability $\bigO(W^{-1})$.  Thus, the conclusions of Proposition~\ref{prop:l1-stable} are in force.
\end{cor}

\bibliographystyle{IEEEbib}
\footnotesize
\bibliography{random-modulator}

\end{document}